\documentclass[10pt,journal]{IEEEtran}

\makeatletter
\def\ifundefined{\@ifundefined}
\makeatother

\usepackage{epsfig}
\usepackage{amsmath}
\usepackage{amssymb} 
\usepackage{graphicx}
\usepackage{amsthm}
\usepackage{mathrsfs}
\usepackage{bbold}
\usepackage[ruled]{algorithm2e}
\usepackage{multirow}
\usepackage{cite}

\newcommand{\bitm}{\begin{itemize}}
\newcommand{\eitm}{\end{itemize}}
\newcommand{\be}{\begin{equation}}
\newcommand{\ee}{\end{equation}}
\newcommand{\bea}{\begin{eqnarray}}
\newcommand{\eea}{\end{eqnarray}}
\newcommand\ba[1]{\left[ \begin{array}{#1}}
\def\ea{\end{array}\right]}
\def\nn{\nonumber\\}

\newcommand{\bfi}{\begin{figure}}
\newcommand{\efi}{\end{figure}}

\newcommand{\mat}[1]{\begin{bf} #1 \end{bf}}
\newcommand{\Mat}[1]{\mat{#1}}
\renewcommand{\Pr}[1]{\mbox{\rm{Pr}}\left\{ #1 \right\}}

\newcommand{\Ev}[1]{\left\{ #1 \right\}}

\newcommand{\Q}{\Mat{Q}}
\renewcommand{\kappa}{n}
\newcommand{\Del}{\pmb{\Lambda}}

\newcommand{\ISI}{\ell}

\DeclareMathOperator{\diag}{diag}
\DeclareMathOperator{\rank}{rank}

\newcommand{\define}{\stackrel{\triangle}{=}}
\newcommand{\rand}[1]{#1}
\newcommand{\D}{\Delta}
\renewcommand{\zeta}{\Omega}

\newcommand{\Real}{\mathbb{R}}
\newcommand{\F}{F}
\newcommand{\I}{\Mat{I}}
\newcommand{\bI}[1]{^{[#1]}}

\newcommand{\BigMin}[1]{\mathop{\min_{\mat{#1} \in \Sym}}_{#1_0 \neq \rand{B}_t}}
\newcommand{\BigMax}[2]{\mathop{\max_{\mat{#1} \in \Sym}}_{#1_0 #2 \rand{B}_t}}
\DeclareMathAlphabet{\mathbbb}{U}{bbold}{m}{n}

\newcommand{\E}{\mathbb{E}}

\newcommand{\Var}{\mbox{Cov}}

\renewcommand{\S}{\Mat{S}}
\newcommand{\s}{\mat{s}}
\newcommand{\PSE}{\Pr{\bigcap_{i=1}^n \Ev{\rand{B}_{t_i}  \neq \rand{A}_{t_i}}}}

\newcommand{\muY}{\pmb{\mu}}
\newcommand{\nuX}{\pmb{\nu}}
\newcommand{\G}{\mat{G}}

\renewcommand{\d}{\delta}
\newcommand{\dmat}{\pmb{\d}}

\newcommand{\U}{\pmb{\rand{U}}}

\newcommand{\Gam}{\Gamma}
\newcommand{\m}{\eta}

\newcommand{\Kw}{\Mat{K}_{\pmb{\rand{W}}}}
\newcommand{\Kv}{\Mat{K}_{\pmb{\rand{V}}}}
\newcommand{\Iv}{^{\dagger}}

\newcommand{\Atk}{\pmb{\rand{A}}_{\mat{t}_1^\kappa}}
\newcommand{\Wtk}{\pmb{\rand{W}}_{\mat{t}_1^\kappa}}
\newcommand{\Rtk}{\pmb{\rand{R}}_{\mat{t}_1^\kappa}}
\newcommand{\Xtk}{\pmb{\rand{X}}_{\mat{t}_1^\kappa}}
\newcommand{\Ytk}{\pmb{\rand{Y}}_{\mat{t}_1^\kappa}}
\newcommand{\Vtk}{\pmb{\rand{V}}_{\mat{t}_1^\kappa}}

\newcommand{\atk}{\mat{a}_1^\kappa}
\newcommand{\ati}{\mat{a}}

\newcommand{\sig}{\sigma}
\newcommand{\randb}[1]{\pmb{\rand{#1}}}
\newcommand{\Alp}{\pmb{\alpha}}
\newcommand{\Beta}{\pmb{\beta}}
\newcommand{\Sym}{\mathcal{M}}
\newcommand{\alp}{\alpha}

\renewcommand{\a}{a}
\newcommand{\at}{a'}
\renewcommand{\L}{m}
\newcommand{\gm}{\mat{g}}
\newcommand{\hO}{\mat{h}_0}

\newcommand{\SetLine}{\SetAlgoLined}
\newcommand{\linesnumbered}{\LinesNumbered}
\newcommand{\nocaptionofalgo}{\NoCaptionOfAlgo}

\newcommand{\1}{\mathbbb{1}}
\newcommand{\0}{\mathbbb{0}}
\renewcommand{\vec}[1]{{\mat{#1}}}

\newcommand{\fn}{}
\newtheorem{thm}{Theorem}  
\newtheorem{lem}{Lemma}        
\newtheorem{cor}{Corollary}

\newtheorem{pro}{Proposition}

\newtheorem{defn}{Definition}       
\newtheorem{rem}{Remark}       

\begin{document}

\title{Reliability Distributions of Truncated Max-log-map (MLM) Detectors Applied to Binary ISI Channels}

\author{\IEEEauthorblockN{Fabian Lim and Aleksandar Kav\v{c}i\'c}
\thanks{F. Lim is with the Research Laboratory of Electronics, Massachusetts Institute of Technology, 77 Massachusetts Ave, Cambridge, MA 02139, USA (e-mail: flim@mit.edu).}
\thanks{A. Kav\v{c}i\'c is with the Department of Electrical Engineering, University of Hawaii at Manoa, 2540 Dole Street, Honolulu, HI 96822, USA. (e-mail: alek@hawaii.edu).}
\thanks{This work was performed when F.~Lim was at the University of Hawaii, and supported by the NSF under grant number CCF-1018984. Parts of this work has been presented at the IEEE International Symposium on Information Theory, St Petersburg, Russia, July 2011}
}

\maketitle

\begin{abstract}
The max-log-map (MLM) receiver is an approximated version of the well-known, Bahl-Cocke-Jelinek-Raviv (BCJR) algorithm. The MLM algorithm is attractive due to its implementation simplicity. In practice, sliding-window implementations are preferred, whereby truncated signaling neighborhoods (around each transmission time instant) are considered. In this paper, we consider binary signaling sliding-window MLM receivers, where the MLM detector is truncated to a length-$m$ signaling neighborhood. Here, truncation is used here to ease the burden of analysis. For any number $n$ of chosen times instants, we derive exact expressions for both i) the \emph{joint} distribution of the MLM symbol reliabilities, and ii) the \emph{joint} probability of the erroneous MLM symbol detections.

We show that the obtained expressions can be efficiently evaluated using Monte-Carlo techniques. The most computationally expensive operation (in each Monte-Carlo trial) is an eigenvalue decomposition of a size $2mn$ by $2mn$ matrix. 
The proposed method handles various scenarios such as correlated noise distributions, modulation coding, etc.
\end{abstract}

\markboth{Lim \MakeLowercase{\textit{et al.}}: {Reliability Distributions of Truncated Max-Log-Map (MLM) Decoder applied to Binary ISI channels}}{}
{
\renewcommand{\v}{\rand{V}}

\section{Introduction}

The intersymbol interference (ISI) channel has been widely studied. Optimal detection schemes for the ISI channel, consider input-output sequences, rather than individual symbols~\cite{Forney}. Sequence detectors such as the \emph{Viterbi} detector, only compute hard decisions~\cite{Viterbi}. However, modern coding techniques often benefit from detection schemes that also compute \emph{symbol reliabilities}, also known as \emph{soft-outputs}, \emph{log-likelihood ratios}~\cite{Turbo,Ma,Lim}. Some well-known detectors that perform this task include the \emph{soft-output Viterbi algorithm} (SOVA)~\cite{SOVA}, the \emph{Bahl-Cocke-Jelinek-Raviv (BCJR) algorithm}~\cite{BCJR}, and the \emph{max-log-map} (MLM) detector~\cite{EquivSOVA}. There has been recent interest in the analysis of the MLM detector. The marginal symbol error probability has been derived for a 2-state convolutional code in~\cite{Yoshi}; this was further extended for convolutional codes with constraint length two in~\cite{Lent}. Approximations for the MLM reliability distributions have been obtained in~\cite{Reggiani,Avu}. 

In this paper, we consider the use of an MLM receiver for binary signaling over an ISI channel. In particular we consider its \emph{sliding-window} implementation. A MLM receiver is termed to be $\L$-truncated, if it only considers a signaling window of length $\L$ around the time instant of interest. The $m$-truncation is used here to falicitate the analysis of the MLM receiver. For the $m$-truncated MLM receiver, considering any number $n$ of chosen time instants, we derive  \emph{exact, closed-form} expressions for \emph{both} i) the \emph{joint} distribution of the symbol reliabilities, and ii) the \emph{joint} probability that the detected symbols are in error. While past work considered only marginal distributions, we provide analytic expressions for joint MLM receiver statistics. Our derivation follows from a simple observation.

\textbf{Notation}: Bold fonts are used to distinguish both vectors and matrices (e.g., denoted~$\mat{a}$ and $ \Mat{A}$, respectively) from scalar quantities (e.g., denoted~$a$). Next, {random} quantities are denoted as follows. Scalars are denoted using upper-case italics (e.g.,, denoted $A$) and vectors denoted using upper-case bold italics (e.g., denoted~$\randb{A}$). Note that we do not reserve specific notation for random matrices. Throughout the paper both $t$ and $\tau$ are used to denote time indices. Sets are denoted using curly braces, e.g., $\{a_1, a_2, a_3, \cdots\}$. Also, both $\alp$ and $\beta$ are used for auxiliary notation as needed. Finally, the maximization over the components of the size-$n$ vector $\mat{a} = [a_1,a_2,\cdots, a_n]^T$, may be written either explicitly as $\max_{i \in \{1,2,\cdots, n \}} a_i$, or concisely as $\max \mat{a}$. Events are denoted in curly brackets, e.g., $\Ev{A\leq a}$ is the event where $A $ is at most $a$. The probability of the event $\Ev{A\leq a}$ is denoted $\Pr{A\leq a}$. The letter $F$ is reserved to denote probability \emph{cumulative} distribution functions, i.e., $F_A(a) = \Pr{A \leq a}$. The expectation of $A$ is denoted as $\E\{A\}$.

\section{The MLM Algorithm}
\begin{figure*}[!t]
	\centering
		\includegraphics[width=.7\linewidth]{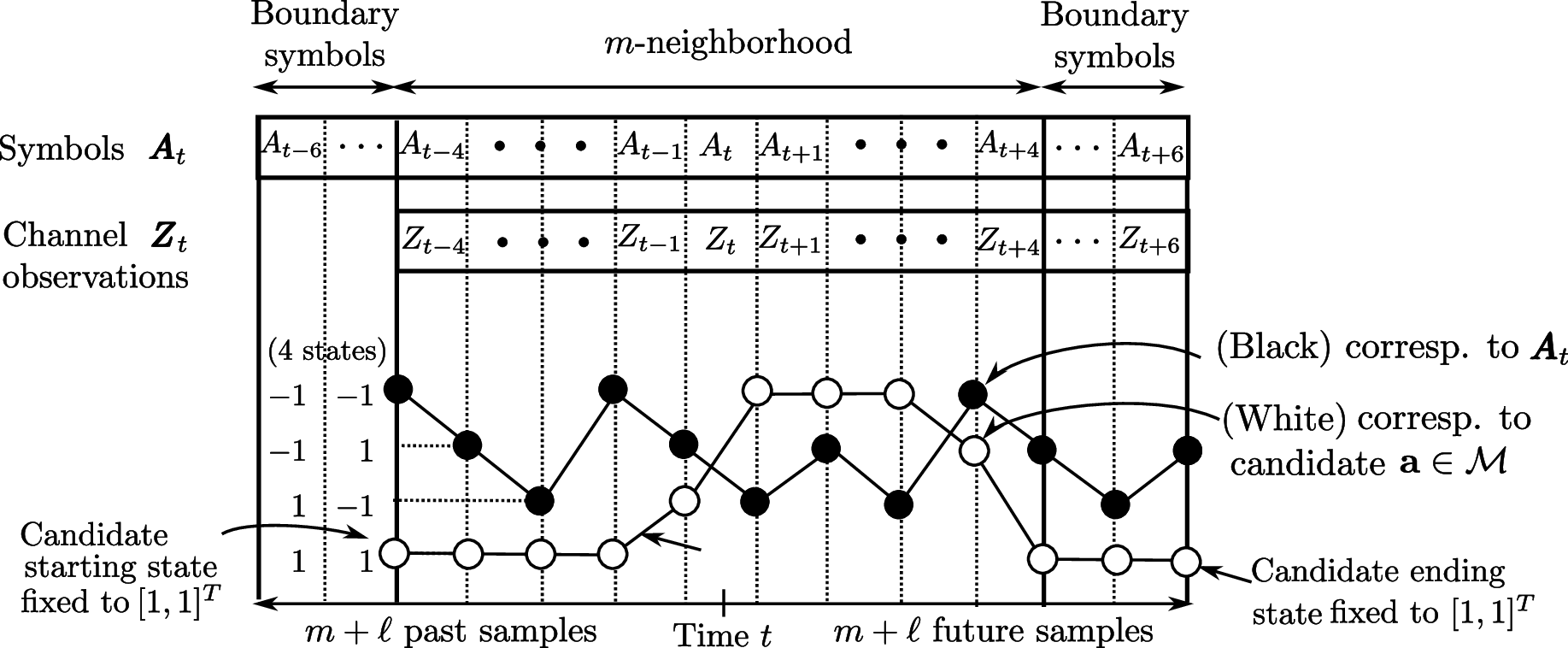}

	\caption{The time evolution of the channel states of an $m$-truncated MLM detector for binary signaling over an ISI channel with memory length $\ell$, where $m = 6$ and $\ell = 2$.}
	\label{fig:Trellis}
		\vspace*{-10pt}
\end{figure*}

Let a random sequence of symbols drawn from the set $\{-1, 1\}$, denoted as $\cdots, \rand{A}_{-2},\rand{A}_{-1}, \rand{A}_{0},\rand{A}_{1}, \rand{A}_{2}, \linebreak[1] \cdots$, be transmitted over an ISI channel with memory $\ell$ characterized by \emph{channel coefficients} $h_0, h_1, \cdots, h_\ISI$. The binary signaling ISI \emph{channel output} sequence, denoted $\cdots, \rand{Z}_{-2}, \rand{Z}_{-1}, \rand{Z}_{0}, \linebreak[1]\rand{Z}_{1}, \rand{Z}_{2}, \cdots$, satisfies the following input-output relationship\vspace*{-7pt}
\bea
\rand{Z}_t \!\! &=& \!\!\sum_{i=0}^\ISI  h_i \rand{A}_{t-i} - \rand{W}_t, \label{eqn:chan1}
\eea
and the channel noise samples $\cdots, \rand{W}_{-2},\rand{W}_{-1}, \rand{W}_{0}, \linebreak[1] \rand{W}_{1}, \rand{W}_{2}, \cdots$ are assumed to be zero-mean and jointly Gaussian distributed (we \emph{do not} assume they are \emph{independent}). Note that the Gaussian noise sample $\rand{W}_t$ in (\ref{eqn:chan1}) is subtracted (as opposed to being usually added in the literature) for purposes of obtaining neater expressions in the sequel. Clearly subtraction incurs no loss in generality, as the Gaussian distribution is symmetric about its mean. The ISI \emph{channel state} at time $t$ equals the (length-$\ell$) vector of input symbols $[\rand{A}_{t-\ell+1},\rand{A}_{t-\ell+2},\cdots, \rand{A}_{t}]^T$. The total number of possible states is $2^\ell$, exponential in the memory length $\ell$.

At time instant $t$, the $\L$-truncated MLM detector considers the neighborhood of $2\L+\ISI+1$ channel outputs $\randb{Z}_t \define [\rand{Z}_{t-\L},\rand{Z}_{t-\L+1},\cdots, \rand{Z}_{t+\L+\ISI}]^T$. Let $\randb{A}_t$ denote the symbol neighborhood that contains the following $2(\L+\ISI)+1$ input symbols
\bea
	\randb{A}_t \define [\rand{A}_{t-\L-\ISI}, \rand{A}_{t-\L-\ISI+1},\cdots, \rand{A}_{t+\L+\ISI}]^T. \label{eqn:sig_vect}
\eea
Both $\randb{A}_t$ and $\randb{Z}_t$ are depicted in Figure \ref{fig:Trellis}. 
Let $\0_{i,j}$ denote a matrix of size $i$ by $j$, whereby every entry of $\0_{i,j}$ equals zero. 
Let $\mat{h}_i$ denote the following length-$(2\L+\ISI+1)$ vector\vspace*{-5pt}
\bea
	\mat{h}_i &\define& [\0_{1,\L+i},h_0,h_1, \cdots, h_\ISI, \0_{1,\L-i}]^T, \label{eqn:h_i}
\eea
where $i$ can take values $|i| \leq m$, note that $\0_{1,i}$ is a row vector of length $i$ containing only zeros. Let both $\Mat{H}$ and $\Mat{T}$ denote the size $2\L+\ell+1$ by $2(\L+\ell)+1$ matrices given as
\newcommand{\Zo}{\mbox{ $\0$  } }
\bea
	\mat{H} \!\!\!\!\! &\define& \!\!\![\0_{2\L+\ISI+1,\ISI}, {\mat{h}_{-\L},\mat{h}_{-\L+1},\cdots,\mat{h}_{\L}},\0_{2\L+\ISI+1,\ISI}], \nn
	\mat{T} \!\!\!\!\! &\define& \!\!\!  
	\left[\begin{array}{ccc} \mbox{\quad\quad$\mat{T}_1$\quad\;} \\ \0_{2\L+1,\ISI} \end{array}, \right.
	\0_{2\L+\ISI+1,2\L + 1}, 
		\left.\begin{array}{ccc} \0_{2\L+1,\ISI}\\\mbox{\quad\quad$\mat{T}_2$\quad\;}  \end{array}, \right]
			   \label{eqn:HandT}
\eea
where the two $\ell$ by $\ell$ \emph{submatrices} $\mat{T}_1$ and $\mat{T}_2$ equal
\bea
\mat{T}_1 \!\!\!\! &=& \!\!\!\!\!\!
		\left[
					\begin{array}{@{}c@{}} 
						\begin{array}{*{12}{@{\hspace{.5ex}}c@{ \hspace{.5ex} }}}
							h_\ISI & h_{\ISI-1} & \cdots & h_1 \\
							&  h_\ISI &  & \vdots \\
							&         &  \ddots &  \vdots  \\
							&         &   & h_\ISI \\ 
						\end{array} \\ 
					\end{array}
		\right],
		\mat{T}_2 = \!\!
		\left[
					\begin{array}{@{}c@{}} 
						
						\begin{array}{*{12}{@{\hspace{.5ex}}c@{ \hspace{.5ex} }}}
							h_0 \\
							\vdots & \ddots \\
							h_{\ISI-2} & \cdots & h_0 \\
							h_{\ISI-1} & \cdots & h_1 & h_0
						\end{array} 
					\end{array}
		\right]  \nonumber
\eea

Using (\ref{eqn:HandT}), rewrite $\randb{Z}_t \define [Z_{t-m},Z_{t-m+1}, \cdots, Z_{t+m+\ISI}]^T$ using (\ref{eqn:chan1}) into the following form
\bea
	\randb{Z}_t &=& \left( \Mat{H} + \Mat{T} \right) \randb{A}_t - \randb{W}_t, \label{eqn:isi_chan}
\eea
where $\randb{W}_t$ denotes the neighborhood of noise samples 
\bea
	\randb{W}_t \define [\rand{W}_{t-\L},\rand{W}_{t-\L+1},\cdots,\rand{W}_{t+\L+\ISI}]^T \label{eqn:Wt}
\eea

Let $\1_\ell$ denote a vector of length $\ell$, with all its entries equal to 1. Let $\Sym$ denote the set of $\L$-truncated MLM {candidate} sequences, defined as
\bea
\!\!\!\!\!\!	\Sym \!\!\!\! &\define& \!\!\!\!\left\{ \mat{\a} \in \{-1,1\}^{2(\L+\ISI)+1}   : \a_i = 1 \mbox{ for all } |i| > \L \right\}.   \label{eqn:Sym}
\eea
Each candidate $\mat{\a} \in \Sym$ has boundary symbols equal to $1$, i.e., each $\mat{\a}$ has the form  \vspace*{-7pt}
\[
\mat{\a} = [\1_\ell^T, \a_{-\L}, \a_{-\L+1},\cdots, \a_{\L}, \1_\ell^T]^T.
\]
Alternatively, the boundary symbols can be specified to be any sequence of choice in the set $\{-1,1\}^\ell$; here we choose $\1_\ell$ for the boundary sequence to simplify exposition. An example of a candidate $\mat{a} \in \Sym$ is illustrated in Figure \ref{fig:Trellis}. The reason for fixing the boundary symbols of the candidates $\mat{a} \in \Sym$ \emph{a priori} (to some chosen sequence), is so as to initialize them to some values (as the boundary symbols of the transmitted sequence $\randb{A}_t$ are \emph{unknown} to the detector). The start/end states of $\randb{A}_t$ (colored \emph{black}), are shown (see Figure \ref{fig:Trellis}) to be different from the start/end states of the candidate $\mat{a} \in \Sym$ (colored \emph{white}). 

\renewcommand{\at}{\bar{\mat{a}}}

Let the sequence $\cdots, \rand{B}_{-2},\rand{B}_{-1}, \linebreak[1] \rand{B}_{0},\rand{B}_{1}, \rand{B}_{2}, \cdots$ denote \emph{symbol decisions} on the channel inputs $\cdots, \rand{A}_{-2},\rand{A}_{-1}, \rand{A}_{0},\rand{A}_{1}, \rand{A}_{2}, \cdots$. Let $\1$ denotes an all-ones vector of non-specified length. In the following let $|\mat{a}|$ denote the Euclidean norm of the vector $\mat{a}$. For each time instant $t$, define the sequence $\randb{B}\bI{t}$ as \vspace*{-5pt}
\bea
\randb{B}\bI{t} &\define& \arg \min_{\mat{\a}\in\Sym} |\pmb{\rand{Z}}_t  - (\Mat{H} + \Mat{T}) \mat{\a} |^2, \nn
&=& \arg \min_{\mat{\a}\in\Sym} |\pmb{\rand{Z}}_t  -  \Mat{T}\1 - \Mat{H}\mat{\a}|^2. \label{eqn:B1}
\eea
The symbol decision $\rand{B}_{t}$ on channel input $\rand{A}_{t}$ is obtained from $\randb{B}\bI{t}$ by setting $\rand{B}_{t} \define \rand{B}\bI{t}_0$, where $\rand{B}\bI{t}_0$ denotes the $0$-th element of the candidate $\randb{B}\bI{t} \in \Sym$. Clearly $\randb{B}\bI{t}$ only has length $2(m+\ell)+1$ and therefore does not equal the MLM bit detection sequence $\cdots, \rand{B}_{-2},\rand{B}_{-1},  \rand{B}_{0},\rand{B}_{1},\linebreak[1]  \rand{B}_{2}, \cdots$; however, note that each symbol $B_t$ is obtained from each $\randb{B}\bI{t}$. Each sequence $\randb{B}\bI{t}$ is obtained by comparing the squared Euclidean distances of each candidate $\Mat{H}\mat{\a}$ from the received neighborhood $\pmb{\rand{Z}}_t -  \Mat{T}\1$, see (\ref{eqn:B1}).

\renewcommand{\fn}{\footnote{The relaxation of this assumption is discussed in the latter-half of the upcoming Subsection \ref{ssect:33}, where we allow some of the probabilities $\Pr{\randb{A}_t=\ati}$ to equal zero, i.e., in the case of modulation coding.}}

In addition to computing \emph{hard}, i.e., $\{-1,1\}$, symbol decisions $B_t$, the $\L$-truncated MLM also computes a symbol \emph{reliability} sequence $\cdots, \rand{R}_{-2},\linebreak[1]\rand{R}_{-1},\linebreak[1] \rand{R}_{0},\rand{R}_{1},\linebreak[1] \rand{R}_{2},\cdots$. Consider the following log-likelihood approximation (see~\cite{EquivSOVA})
\bea
	\log \frac{\Pr{A_t = B_t | \randb{Z}_t}}{\Pr{A_t \neq B_t| \randb{Z}_t}} \!\!\!\!
   &=& \!\!\!\!
	\log \frac{\sum_{\mat{a} \in \Sym : \; a_0 = B_t} \Pr{\randb{Z}_t|A_t = \mat{a} }}{\sum_{\mat{a} \in \Sym : \; a_0 \neq B_t} \Pr{\randb{Z}_t| A_t = \mat{a}}  } \nn 
	&\approx& \!\!\!\!\!\!  \mathop{\min_{\mat{a} \in \Sym}}_{a_0 \neq \rand{B}_t} \frac{1}{2\sig^2} |\pmb{\rand{Z}}_t - \Mat{T}\1 - \Mat{H}\vec{\a} |^2 \nn
	&& \!\!\!\!  - \mathop{\min_{\mat{a} \in \Sym}}_{a_0 = \rand{B}_t} \frac{1}{2\sig^2}  |\pmb{\rand{Z}}_t - \Mat{T}\1 - \Mat{H}\vec{\a} |^2, \label{eqn:llr}
\eea
where the first equality assumes\fn~uniform signal priors , i.e. $\Pr{\randb{A}_t=\ati}=2^{-2(m+\ell)-1}$, see (\ref{eqn:sig_vect}). Denote $\sig^2$ to be the worst-case noise variance
\bea
 \sig^2 & \define & \sup_{t \in \mathbb{Z}} \E \{W_t^2\}, \label{eqn:snr}
\eea 
and assume that $\sig^2$ is bounded, i.e., $\sig^2< \infty$. If $\rand{W}_t$ is stationary, then $\sig^2 = \E\{\rand{W}_t^2 \}$. We want to set the ($m$-truncated MLM) reliability $R_t$ to equal the log-likelihood approximation (\ref{eqn:llr}), written in the following form. Denote the difference in the obtained squared Euclidean distances 
\bea
 \D(\vec{\a},\vec{\at}) &=&\D(\mat{\a},\mat{\at}; \randb{Z}_t)   \nn
  &\define& \!\!\!\!\! |\pmb{\rand{Z}}_t - \Mat{T}\1 - \Mat{H}\vec{\a} |^2 - |\pmb{\rand{Z}}_t - \Mat{T}\1 - \Mat{H}\vec{\at} |^2 ,\label{eqn:D}
\eea
where both $\mat{\a}$ and $ \mat{\at} $ are arbitrary sequences in $\{-1,1\}^{2(\L+\ISI)+1}$. Recalling (\ref{eqn:B1}), we write $R_t$ as follows.

\begin{defn} \label{defn:relt}
The non-negative $\L$-truncated MLM reliability $\rand{R}_t$ is defined as
\bea
\rand{R}_t &\define& \BigMin{\a} \frac{1}{2\sig^2} \D(\vec{\a},\pmb{\rand{B}}\bI{t}) , \label{relt}
\eea
where $\D(\vec{\a},\pmb{\rand{B}}\bI{t}) \geq 0$, is the difference in the obtained squared Euclidean distances corresponding to candidates $\mat{a},\pmb{\rand{B}}\bI{t} \in \Sym$, and $ \sig^2 $ is the noise variance (\ref{eqn:snr}).
\end{defn}

Note that $\D(\vec{\a},\pmb{\rand{B}}\bI{t}) \geq 0$ for all $\mat{a} \in \Sym$, simply because $\pmb{\rand{B}}\bI{t}$ achieves the minimum squared Euclidean distance amongst all candidates in $\Sym$, see (\ref{eqn:B1}).

\section{Key Observation and Main Result}

In the first subsection, we describe an important \emph{key observation}, of which the derivation of the main result is based on. In the second subsection, the main result provides closed-form expressions for the i) joint reliability distribution $F_{R_{t_1}, R_{t_2},\cdots, R_{t_n}}(r_1,r_2,\cdots, r_n)$, and ii) joint symbol error probability $\PSE$, for $n$ time instants $t_i$. A Monte-Carlo procedure is give to evaluate these closed-forms. In the third subsection, an efficient method to run the Monte-Carlo is discussed.

\subsection{Key observation}

For all $t$, define $\rand{X}_t$ and $\rand{Y}_t$ as
\bea
\rand{X}_t &\define& \mathop{\max_{\mat{a} \in \Sym}}_{a_0 \neq \rand{A}_t} \frac{1}{4} \D(\pmb{\rand{A}}_t,\vec{\a}), \nn
\rand{Y}_t &\define& \mathop{\max_{\mat{a} \in \Sym}}_{a_0 = \rand{A}_t} \frac{1}{4}  \D(\pmb{\rand{A}}_t,\vec{\a})  \geq 0 , \label{eqn:XaY}
\eea
where $\D(\randb{A}_t,\mat{a})$, see (\ref{eqn:D}) equals the difference of, the squared Euclidean distances corresponding to $\randb{A}_t$, and a candidate $\mat{a} \in \Sym$, respectively. Note that $\rand{Y}_t \geq 0 $, because there must exist a candidate $\mat{\a} \in \Sym$ that satisfies $\D(\pmb{\rand{A}}_t,\vec{\a}) =0$, see (\ref{eqn:D}); this particular candidate $\mat{a} \in \Sym$ satisfies $a_i= A_{t+i}$ for all values of $i$ satisfying $|i| \leq m$.

\begin{pro}[Key Observation] \label{relprop}
The $\L$-truncated MLM reliability $\rand{R}_t$ in (\ref{relt}) satisfies
\bea
\rand{R}_t &=& \frac{2}{\sig^2}|\rand{X}_t-\rand{Y}_t|, \label{eqn:relprop}
\eea
where both random variables $\rand{X}_t$ and $ \rand{Y}_t$ are given in (\ref{eqn:XaY}). \hspace*{\fill}\IEEEQEDopen
\begin{proof} \rm
Scale (\ref{relt}) by $\sig^2/2$ and write
\bea
\frac{\sig^2}{2} \cdot \rand{R}_t \!\! &=& \BigMin{a} \frac{\D(\vec{a},\pmb{\rand{A}}_t)}{4}  + \frac{\D(\pmb{\rand{A}}_t,\pmb{\rand{B}}\bI{t})}{4} \nn
&=& \!\!\!\! \left(- \BigMax{a}{\neq} \frac{\D(\pmb{\rand{A}}_t,\vec{a})}{4} \right)  + \frac{\D(\pmb{\rand{A}}_t,\pmb{\rand{B}}\bI{t})}{4}. \label{relt2}
\eea
To obtain the last equality in (\ref{relt2}), we used the relationship $\D(\randb{A}_t,\mat{\a}) = - \D(\vec{\a},\randb{A}_t)$, see (\ref{eqn:D}). Recall the symbol decision $B_t \define \rand{B}_0\bI{t}$, where $\randb{B}\bI{t}$ is defined in (\ref{eqn:B1}). Because $\rand{B}_t$ is either $-1$ or $1$, we have either $\rand{B}_t \neq \rand{A}_t$ or $\rand{B}_t = \rand{A}_t$. Consider the former case $\rand{B}_t \neq \rand{A}_t$, in which (\ref{relt2}) reduces to 
\bea
\frac{\sig^2}{2} \cdot  \rand{R}_t 
&=&  \left( - \mathop{\max_{\mat{a} \in \Sym}}_{a_0 = \rand{A}_t} \frac{\D(\pmb{\rand{A}}_t,\vec{a})}{4} \right)
    +  \mathop{\max_{\mat{a} \in \Sym}}_{a_0 \neq \rand{A}_t} \frac{\D(\pmb{\rand{A}}_t,\mat{a})}{4}, \nn
&=&  -\rand{Y}_t  + \rand{X}_t =|\rand{X}_t - \rand{Y}_t|, \nonumber
\eea
where the second equality follows from (\ref{eqn:XaY}), and the third from the fact $\rand{R}_t \geq 0$, see Definition \ref{defn:relt}. We have thus shown (\ref{eqn:relprop}) for the case $\rand{B}_t \neq \rand{A}_t$. The same conclusion follows for the other case $\rand{B}_t = \rand{A}_t$ in similar manner.
\end{proof}
\end{pro}

The expression (\ref{eqn:relprop}) is developed for purposes of analysis, and cannot be used to compute $R_t$. In practice, the quantities $X_t$ and $Y_t$ cannot never be computed, as they require knowledge of the transmitted sequence $\randb{A}_t$, see (\ref{eqn:XaY}). Such knowledge is never available at the detector, because the detector is in fact trying to estimate $\randb{A}_t$. The simple Proposition \ref{relprop}, which seems completely overlooked in past literature, enables the derivation of the main result.

\subsection{Expressions for joint reliability distribution and symbol error probability: Main result}

For any $\kappa$ number of arbitrarily chosen time instants $t_1, t_2, \cdots, t_\kappa$, we wish to obtain the distribution of the vector $\Rtk$, containing the following reliabilities
\bea
	\Rtk &\define &[\rand{R}_{t_1}, \rand{R}_{t_2},  \cdots,  \rand{R}_{t_\kappa}]^T. \label{eqn:Rtk}
\eea

\newcommand{\GMats}{\Mat{G}(\randb{A}_{t_1}),\Mat{G}(\randb{A}_{t_2}), \cdots, \Mat{G}(\randb{A}_{t_n})}
\newcommand{\bdiag}{\pmb{\diag}}

Recall $\0_{1,i}$ denotes a length $i$ vector with all entries equal to $0$. Define a \emph{binary vector} $\mat{e}_i$ of length $2(\L+\ISI)+1$ as 
\bea
\mat{e}_i &\define& [\0_{1,\L+\ISI+i},  1, \0_{1,\L+\ISI-i}]^T,  \label{eqn:mate}
\eea
where $i$ can take values $|i| \leq \L + \ISI$. Further define the matrix $\Mat{E}$ of size $2(\L+\ISI) + 1$ by $2\L$ as 
\bea
	\Mat{E} \define [\mat{e}_{-\L}, \mat{e}_{-\L+1},\cdots, \mat{e}_{-1},\mat{e}_{1}, \mat{e}_{2}, \cdots, \mat{e}_{\L}]. \label{eqn:Emat}
\eea
Let $\diag(\randb{A}_t)$ denote the \emph{diagonal matrix}, whose diagonal equals the vector $\randb{A}_t$. Define the following $2m + \ell + 1$ by $2^{2m}$ matrix 
\bea
	\G(\randb{A}_t) &\define& \Mat{H} \diag(\randb{A}_t) \Mat{E}, \label{eqn:G}
\eea
where the noise neighborhood $\randb{W}_t$ is given by (\ref{eqn:Wt}). Let $\Wtk$ denote the concatenation 
\bea
\Wtk \define \ba{c} \randb{W}_{t_1} \\ \randb{W}_{t_2} \\ \vdots \\ \randb{W}_{t_n} \ea. \label{eqn:Wtk}
\eea
Define the noise covariance matrix 
\bea
\Kw &\define&
\ba{ccc}
\E\{\pmb{\rand{W}}_{t_1} \pmb{\rand{W}}_{t_1}^T \} & \cdots & \E\{\pmb{\rand{W}}_{t_1}\pmb{ \rand{W}}_{t_n}^T \} \\
 \vdots & \ddots  & \vdots \\
 \E\{\pmb{\rand{W}}_{t_n} \pmb{\rand{W}}_{t_1}^T \} & \cdots & \E\{\pmb{\rand{W}}_{t_n} \pmb{\rand{W}}_{t_n}^T \}
\ea \nn &=& \E\{ \Wtk \Wtk^T\},
 \label{eqn:mx_cov} 
\eea
where note that $\Kw$ is generally not Toeplitz even if $\rand{W}_t$ is stationary. As in (\ref{eqn:Wtk}), let $\Atk$ denote the concatenation 
\bea
\Atk \define \ba{c} \randb{A}_{t_1} \\ \randb{A}_{t_2} \\ \vdots \\ \randb{A}_{t_n} \ea. \label{eqn:Atk}
\eea
Let $\I$ denote the identity matrix; in particular $\I_{2\L}$ has size $2\L \times 2\L$. Define the matrix $\S$ of size $2\L$ by $2^{2\L}$ as 
\bea
	\S \define [\s_0, \s_1, \cdots, \s_{2^{2\L}-1}], \label{eqn:S}
\eea
where the columns $\s_0, \s_1, \cdots, \s_{2^{2\L}-1}$ make up all $2^{2\L}$ possible, length-$(2\L)$ binary vectors, i.e., $\Ev{\s_0, \s_1, \cdots, \s_{2^{2\L}-1}}= \{0,1\}^{2m}$. The matrix $\S\S^T$ has the following simple expression
\bea
	\S\S^T &=& \sum_{k=0}^{2^{2\L}-1} \s_k \s_k^T =  2^{2(\L-1)}  \cdot [\I_{2\L} + \1\1^T],  \label{eqn:SS}
\eea
where here the vector $\1$ has all entries equal to $1$. Denote the matrix \emph{Kronecker product} using the operation $\otimes$. Let $\bdiag\left( \GMats \right)$ denote a \emph{block diagonal matrix}, whose block-diagonal entries equal $\GMats$.

\begin{defn} \label{defn:QD}
Let the square matrix $\Q = \Q(\Atk)$ of size $2mn$ by $2mn$ satisfy the following two conditions:
\bitm
\item[i)]  the matrix $\Q$ {decomposes} 
\bea
  \!\!\!\!  \!\!\!\! \Q\Del^2\Q^T   \!\!\!\!\!\!
&=&   \!\!\!\! \bdiag\left(\GMats\right)^T \Kw  \nn
&&  \!\cdot~\bdiag\left(\GMats\right),
\nn
 \label{eqn:QDQ}
\eea
where 
$\Del =\Del(\Atk)$ is a diagonal matrix. The number of positive diagonal elements in the matrix $\Del$, equals the rank of the matrix (\ref{eqn:QDQ}).  
\item[ii)] the matrix $\Q$ \textbf{diagonalizes} the matrix $\I_n \otimes \S\S^T$, i.e., the matrix $\Q$ satisfies
\bea
	\Q^T (\I_n \otimes \S\S^T) \Q = \I,
\eea
noting that the matrix $\S\S^T$ is square of size $2m$. 
\eitm
\end{defn}

\renewcommand{\fn}{\footnote{The matrix appearing in (\ref{eqn:F}), with elements $\mat{G}(\randb{A}_{t_i})$, can also be written as $\bdiag\left( \GMats \right)$.}}
Appendix \ref{app:spec} describes the computation of $\Q = \Q(\Atk)$, and $\Del = \Del(\Atk)$ in (\ref{eqn:QDQ}). We partition matrix $\Mat{Q}$ into $n$ partitions of equal size $2m \times 2mn$, i.e., 
\bea
\Q = \ba{c} \Q_1 \\ \Q_2 \\ \vdots \\ \Q_n \ea. \label{eqn:Qpart}
\eea
Let $\diag(A_{t_1},A_{t_2},\cdots, A_{t_n}) $ denote the diagonal matrix, whose diagonal equals $[A_{t_1},A_{t_2},\cdots, A_{t_n}]^T$. Define the $n \times 2\L n $ matrix $\Mat{\F}(\Atk)$ as\fn 
\bea
\Mat{F}(\Atk) 
&\define& \diag(A_{t_1},A_{t_2},\cdots, A_{t_n}) \otimes \mat{h}_0^T \Kw \nn
&& \cdot
\renewcommand{\arraystretch}{.7}
\ba{@{}c@{}c@{}c@{}c@{}} \Mat{G}(\randb{A}_{t_1})  \\ & \Mat{G}(\randb{A}_{t_2})  \\ & & \ddots \\ & & & \Mat{G}(\randb{A}_{t_n}) \ea 
\ba{c} \S\S^T\Q_1 \\ \S\S^T\Q_2 \\ \vdots \\ \S\S^T\Q_n \ea 
\Del\Iv, \label{eqn:F} \nn
\eea
where $\mat{h}_0$ is given in (\ref{eqn:h_i}), 
and $\Del\Iv$ is formed by reciprocating only the \emph{non-zero} diagonal elements of $\Del$. Define the following length-$2^{2\L}$ vectors $\muY(\randb{A}_t)$ and $\nuX(\randb{A}_t)$ as
\begin{align}
\muY(\randb{A}_t)  =& [\mu_1, \mu_2, \cdots, \mu_{2^{2m}-1}]^T \nn
\define &    [\Mat{G}(\randb{A}_t)\S]^T \cdot \Mat{T} (\1 - \randb{A}_t) \nn
&   - \left[|\Mat{G}(\randb{A}_t)\s_0|^2, |\Mat{G}(\randb{A}_t)\s_1|^2,\cdots, |\Mat{G}(\randb{A}_t)\s_{2^{2\L}-1}|^2\right]^T, \nn \label{eqn:muY} \\
\nuX(\randb{A}_t)  =& [\nu_1, \nu_2, \cdots, \nu_{2^{2m}-1}]^T  \nn 
\define& \muY(\randb{A}_t) - 2\rand{A}_t \cdot \mat{h}_0 ^T \Mat{G}(\randb{A}_t)\S \label{eqn:nuX},
\end{align}
where $\mu_k = \mu_k(\randb{A}_t)$ and $\nu_k = \nu_k(\randb{A}_t)$ denote the $k$-th components of $ \pmb{\mu}_k(\randb{A}_t)$ and $\pmb{\nu}_k(\randb{A}_t) $ respectively, and $\Mat{T}$ is given in (\ref{eqn:HandT}). Let $\Phi_{\Mat{K}}(\mat{r})$ denote the distribution function of a zero-mean Gaussian random vector with covariance matrix $\Mat{K}$. Finally define the following length-$n$ random vectors 
\bea
	\Xtk &\define &[\rand{X}_{t_1}, \rand{X}_{t_2},  \cdots,  \rand{X}_{t_\kappa}]^T, \nn
	\Ytk &\define &[\rand{Y}_{t_1}, \rand{Y}_{t_2},  \cdots,  \rand{Y}_{t_\kappa}]^T,
\eea
where both $\rand{X}_{t_i}$ and $\rand{Y}_{t_i}$ are given in (\ref{eqn:XaY}). Let $\Real$ denote the set of real numbers. 

\newcommand{\Om}{\U}

\begin{thm} \label{thm:1}
Define the following quantities:
\bitm

\item let $\Om $ denote a standard zero-mean identity-covariance Gaussian random vector of length-$(2\L\kappa)$. 

\item let $\dmat(\Om,\Atk)=[\d_1,\d_2,\cdots,\d_\kappa]^T$ denote a length-$n$ vector in $\Real^n$, that satisfies
\bea
 \!\!\!\!\!\!\!\! \d_i = \d_i(\Om, \Atk) 
 \!\!\!\!  &\define&  \!\!\!\!  \max( \S^T\Q_i\Del\Om + \muY(\randb{A}_{t_i})) \nn && - \max ( \S^T\Q_i\Del\Om + \nuX(\randb{A}_{t_i})).  \label{eqn:m}
\eea

\item let $\pmb{\eta}(\Om,\Atk) = [\eta_1,\eta_2,\cdots,\eta_\kappa]^T$  denote a length-$n$ vector in $\Real^n$, that satisfies
\bea
\pmb{\eta}(\Om,\Atk)  
&\define& \diag(\rand{A}_{t_1},\rand{A}_{t_2},\cdots, \rand{A}_{t_n}) 
\Mat{T} \nn 
&& \cdot \left(\1\cdot \1^T - [\randb{A}_{t_1},\randb{A}_{t_2},\cdots, \randb{A}_{t_n}]\right)^T \mat{h}_0   \nn
&& -|\mat{h}_0 |^2 \cdot \1  + \Mat{F}(\Atk) \Om.  \label{eqn:eta}
\eea

\item let $\Kv(\Atk)$ denote a $n \times n$ matrix as follows
\bea
\Kv(\Atk) &\define& \diag(\rand{A}_{t_1},\rand{A}_{t_2},\cdots, \rand{A}_{t_n}) \otimes \mat{h}_0^T \Kw \nn
&&\cdot \diag(\rand{A}_{t_1},\rand{A}_{t_2},\cdots, \rand{A}_{t_n}) \otimes \mat{h}_0 \nn
&& - \Mat{F}(\Atk) \Mat{F}(\Atk)^T. \label{eqn:nu}
\eea
\eitm
Then the distribution of $\Xtk-\Ytk$ is given as
\bea
F_{\Xtk-\Ytk}(\mat{r})   \!\!\!\!  &=&   \!\!\!\! \E\left\{ 
  \Phi_{\Kv(\Atk)}\left(\mat{r} + \dmat(\Om, \Atk) - \pmb{\eta}(\Om, \Atk)\right) 
\right\}
 \nn \label{eqn:main_preview}
\eea
for all $\mat{r} \in \Real^\kappa$.
\end{thm}

The proof of Theorem \ref{thm:1} is given in Section \ref{sect:dens}. Both i) the joint distribution of the reliabilities $\Rtk \define [\rand{R}_{t_1},\rand{R}_{t_1},\cdots, \rand{R}_{t_n}]^T$ in (\ref{eqn:Rtk}), and ii) the joint error probability $\PSE$, follow as corollaries from Theorem \ref{thm:1}. In the following we denote an index subset $\{\tau_1,\tau_2,\cdots, \tau_j \} \subseteq \{t_1,t_2,\cdots,t_n\}$ of size $j$, written compactly in vector form as $\pmb{\tau}_1^j=[\tau_1,\tau_2,\cdots, \tau_j]^T$. 

\begin{cor} \label{cor:main}
The distribution of $\Rtk = 2/\sig^2 \cdot |\Xtk - \Ytk|$, see Proposition \ref{relprop}, is given by
\begin{align}
F_{\Rtk}(\mat{r})  &= F_{|\Xtk - \Ytk|}(\sig^2 /2 \cdot \mat{r})  \nn
&= \sum_{j=0}^n\mathop{\sum_{\{\tau_1,\tau_2,\cdots,\tau_j \} \subseteq }}_{\{t_1,t_2,\cdots,t_n \}}\!\!\!\!\!\!\!\! (-1)^j \cdot
  F_{\Xtk - \Ytk}\left(\frac{\sig^2}{2} \cdot \Alp(\pmb{\tau}_1^j, \mat{r}) \right) \nonumber
\end{align}
where the length-$n$ vector $\Alp(\pmb{\tau}_1^j, \mat{r})=[\alp_1,\alp_2,\cdots,\alp_n]^T$ satisfies
\[
\alp_i =\alp_i(\pmb{\tau}_1^j,r_i)  = \left\{\begin{array}{rl}
		-r_i & \mbox{ if } t_i \in \{\tau_1,\tau_2,\cdots, \tau_j\},  \\
		r_i & \mbox{ otherwise },
			\end{array} \right.
\]
and $F_{\Xtk - \Ytk}\left(\frac{\sig^2}{2}\cdot \Alp(\pmb{\tau}_1^j, \mat{r})\right)$ has the similar closed form as in Theorem \ref{thm:1}. \hspace*{\fill}\IEEEQEDopen
\end{cor}

\newcommand{\ttt}{{\mat{t}_1^{n-1}}}

Corollary \ref{cor:main} can be verified using recursion; for the $n$-th case we express 
\bea
F_{|\Xtk - \Ytk|}(\mat{r}) &=&  F_{|\randb{X}_\ttt - \randb{Y}_\ttt|, X_{t_n} - Y_{t_n}}(\mat{r}_1^{n-1}, r_n) \nn
&& -  F_{|\randb{X}_\ttt - \randb{Y}_\ttt|, X_{t_n} - Y_{t_n}}(\mat{r}_1^{n-1}, -r_n).
\nonumber
\eea
Observe that we still may apply Corollary \ref{cor:main} to each of the two terms on the r.h.s.; we apply Corollary \ref{cor:main} only to the variables $|\randb{X}_\ttt - \randb{Y}_\ttt|$, at the same time accounting for the (respective) joint events $\{X_{t_n} - Y_{t_n} \leq r_n \}$ and $\{X_{t_n} - Y_{t_n} \leq -r_n \}$. The desired expression will be obtained after using some algebraic manipulations.

\newcommand{\om}{\mat{u}}

\begin{algorithm}[!t]
	\SetLine
	\linesnumbered
	\nocaptionofalgo
	\SetKwInput{Init}{Initialize} 
	\Init{Set $F_{\Xtk-\Ytk}(\mat{r}) := 0$ for all $\mat{r} \in \Real^\kappa$;}
	\While{$F_{\Xtk-\Ytk}(\mat{r})$ not converged}{
	  Sample $\Atk=\atk$ using $\Pr{\Atk = \atk}$. Sample the length-$n$, standard zero-mean identity-covariance Gaussian vector $\Om = \om$\;
		Using the sampled realization $\Atk=\atk$, obtain the matrices $\Q = \Q(\atk)$ and $\Del = \Del(\atk)$ satisfying Definition \ref{defn:QD}, see Appendix \ref{app:spec}\;
		Compute $\d_i = \d_i(\om, \atk)$ for all $i \in \{1,2,\cdots, n\}$. For $\d_i$ compute
		\bea
			\max_{k\in\{0,1,\cdots,2^{2\L}-1\}} \s_k^T\Q_i \Del\om + \mu_k(\mat{a}),  \nn
			 \nn	
			\max_{k\in\{0,1,\cdots,2^{2\L}-1\}}  \s_k^T\Q_i \Del\om  + \nu_k(\ati),  \nonumber
			\eea
			see (\ref{eqn:m}). Here $\ati$ is the sampled realization $\randb{A}_{t_i} = \ati$, and both $\mu_k(\ati)$ and $\nu_k(\ati)$ are the $k$-th components of $\muY(\ati)$ and $\nuX(\ati)$, see (\ref{eqn:muY}) and (\ref{eqn:nuX})\;
			Compute $\Mat{F}(\Atk)$ in (\ref{eqn:F}); Also compute $\pmb{\eta}(\om, \atk)$ in (\ref{eqn:eta}) and $\Kv(\atk)$ in (\ref{eqn:nu})\;
		Update 
		\begin{align}
			F&_{\Xtk - \Ytk}(\mat{r}) \nn :=& F_{\Xtk - \Ytk}(\mat{r})  +   \Phi_{\Kv(\atk)} \left( \mat{r} +  \dmat(\om,\atk) -\! \pmb{\eta}(\om, \atk) \right) \nonumber
			\end{align}
			for all $\mat{r} \in \Real^\kappa$\;
	}	
	\caption{\textbf{Procedure 1}: Evaluate Joint Distribution $F_{\pmb{\rand{X}}_{\textbf{t}_1^n}-\pmb{\rand{Y}}_{\textbf{t}_1^n}}(\textbf{r})$}
	\label{proce:pdfXmY}
\end{algorithm}

\begin{cor} \label{cor:main2}
The probability $\PSE$ that \textbf{all} symbol decisions $\rand{B}_{t_1},\rand{B}_{t_2},\cdots, \rand{B}_{t_n}$ are in error, equals
\begin{align}
&\PSE  = \Pr{\Xtk \geq  \Ytk}  \nn
&= 1 + \sum_{j=1}^n \mathop{\sum_{\{\tau_1,\tau_2,\cdots,\tau_j \} \subseteq }}_{\{t_1,t_2,\cdots,t_n \} } (-1)^j \cdot 
   F_{\randb{X}_{\pmb{\tau}_1^j}-\randb{Y}_{\pmb{\tau}_1^j}}(\mat{0}), \nonumber
\end{align}
where the probability 
\bea
F_{\randb{X}_{\pmb{\tau}_1^j}-\randb{Y}_{\pmb{\tau}_1^j}}(\mat{0}) = 
\Pr{\bigcap_{\tau \in \{\tau_1,\tau_2,\cdots, \tau_j\}} \{ X_{\tau} - Y_\tau \leq 0 \}} \nonumber
\eea
has a similar closed form as in Theorem \ref{thm:1}.  \hspace*{\fill}\IEEEQEDopen
\begin{proof} \rm
From (\ref{eqn:XaY}) we clearly see that the event $\Ev{X_t \geq Y_t}$ indicates that the sequence $\pmb{B}\bI{t}$ in (\ref{eqn:B1}) will have its $0$-th component $B\bI{t}_0 \neq A_t$. Because the symbol decision $B_t$ is set to $B_t = {B}\bI{t}_0$ (where $\randb{B}\bI{t}$ is defined in (\ref{eqn:B1})) the event $\Ev{X_t \geq Y_t}$ indicates that $B_t \neq A_t$, which is exactly a symbol decision error occurring at time $t$.
\end{proof}
\end{cor}

Denote the realizations of $\Atk$, $\randb{A}_t$ and $\U$, as $\Atk = \atk$, and $\randb{A}_t = \mat{a}$, and $\U = \mat{u}$. The Monte-Carlo procedure used to evaluate the closed-form of $F_{\Xtk - \Ytk}(\mat{r})$ in Theorem \ref{thm:1}, is given in Procedure \ref{proce:pdfXmY}. We may reduce the number of computations used to obtain matrices $\Q=\Q(\Atk)$ and $\Del=\Del(\Atk)$ in Line 3, by sampling $\U = \mat{u}$ multiple times for a fixed $\Atk = \atk$.

\begin{rem}
The matrix $\Kv(\atk)$ computed in Line 5 (see also (\ref{eqn:nu})) may not have full rank. Hence, when evaluating the Gaussian distribution function  $\Phi_{\Kv(\atk)} (\mat{r})$ with covariance matrix $\Kv(\atk)$ in Line 6, we may require techniques designed for rank deficient covariances, see for example~\cite{Somer}.
\end{rem}

Our proposed method requires no assumptions on the noise covariance matrix $\Kw$ in (\ref{eqn:mx_cov}), and can be applied even when the noise $W_t$ is correlated and/or non-stationary. However, there is an implicit assumption that $\Atk$ is equally-likely amongst all realizations $\Atk=\atk$ that have non-zero probability. Further modifications will be required to extend our method to the general case of non-uniform priors $\Pr{\Atk=\atk}$ (the first equality of (\ref{eqn:llr}) is not valid for such cases).

\begin{rem} \label{rem:conv}
Because $\Phi_{\mat{K}}(\mat{r})$ is a probability distribution function, therefore
\[
0 \leq \Phi_{\Kv(\Atk)} \left( \mat{r} + \dmat(\U,\Atk) - \pmb{\eta}( \U,\Atk) \right) \leq 1,
\]
the well-known Hoeffding probability inequalities can be applied to obtain convergence guarantees, see~\cite{Hoef}. 
\end{rem}

The main thrust of the next subsection is to address Line 4 of Procedure \ref{proce:pdfXmY}. It appears that to execute Line 4 of Procedure \ref{proce:pdfXmY}, we require an exhaustive search over $2^{2\L}$ terms to perform the two maximizations. However, we point out in the next subsection, that these maximizations can be performed more efficiently by utilizing dynamic programming optimization techniques. Also, in the next subsection, we address the computation of $F_{\pmb{X}_{\mathbf{t}_1^n}- \pmb{Y}_{\mathbf{t}_1^n}}(\mathbf{r})$, in instances where one wishes to only consider a subset $\bar{\Sym} \subset \Sym$ (see (\ref{eqn:Sym})).

\subsection{On efficient computation of the closed-form expressions} \label{ssect:33}

\begin{figure}[!t]
	\centering
		\includegraphics[width=\linewidth]{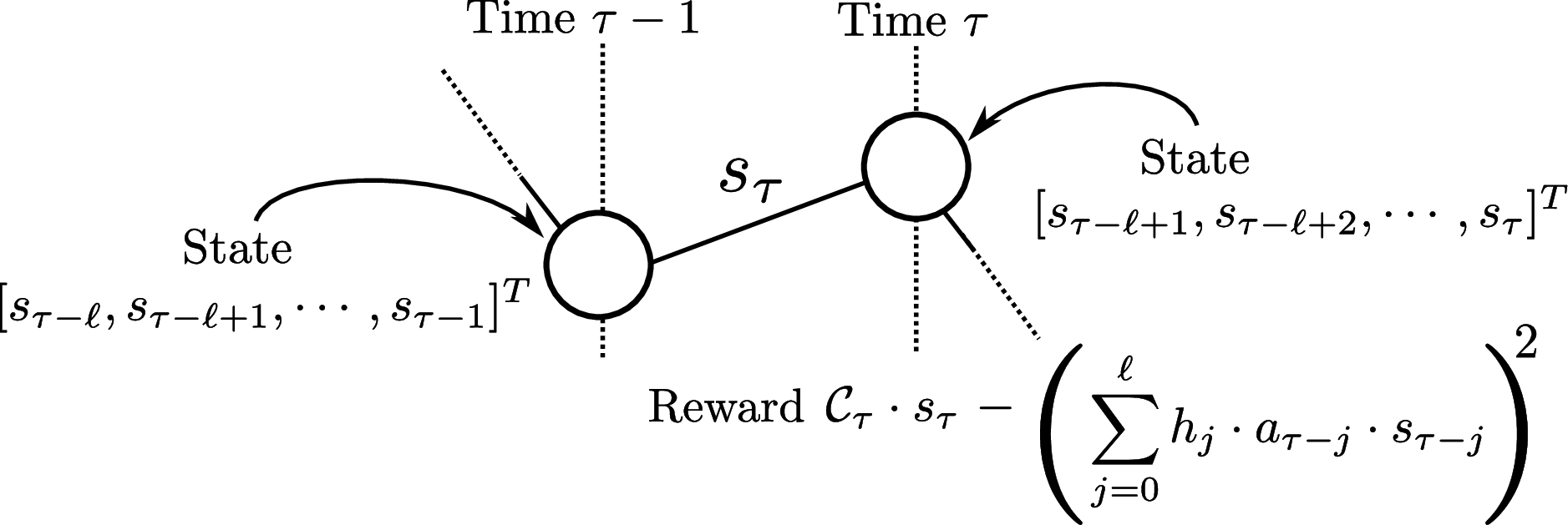}
	
	\caption{Time evolution of the dynamic programming states.}
	\label{fig:CS2}
\end{figure}

\renewcommand{\k}{\tau}

To compute $\d_i$ in (\ref{eqn:m}) while executing Line 4 of Procedure \ref{proce:pdfXmY}, we need to perform the following two maximizations 
\begin{align}
		\max_{\s \in \{0,1\}^{2m}} &\; \s^T \Q_i \Del\om +  [\Mat{G}(\ati)\s]^T \cdot \Mat{T} (\1 - \ati) - |\Mat{G}(\ati)\s|^2, \nn	
		\max_{\s \in \{0,1\}^{2m}} &\; \s^T \Q_i \Del\om +  [\Mat{G}(\ati)\s]^T \cdot [\Mat{T} (\1 - \ati)-2 a_{0} \cdot \mat{h}_0] \nn 
		&- |\Mat{G}(\ati)\s|^2 , \label{eqn:Max}
\end{align}
where both $\ati$ and $\mat{u}$ are realizations $\randb{A}_{t_i}=\ati$ and $\U = \mat{u}$. Note that we obtain (\ref{eqn:Max}) from (\ref{eqn:m}), by substituting for both $\muY(\ati)$ and $\nuX(\ati)$ using (\ref{eqn:muY}) and (\ref{eqn:nuX}) respectively. Index the realization $\randb{A}_{t_i}=\ati$ as 
\[
	\mat{a} \define [a_{-m-\ell}, a_{-m-\ell+1}, \cdots, a_{m+\ell}]^T.
\]
Let $\diag(\ati)$ denote the diagonal matrix, with diagonal $\mat{a}$. Let $\gm_\k$ denote the length $2(\L+\ISI) + 1$ vector
\begin{align}
	&\mat{g}_\k \define  [\0_{1,\L+\k}^T,h_\ISI  a_{\k-\ISI} ,h_{\ISI-1}  a_{\k-(\ISI-1)}, \cdots, h_0  a_\k, \0_{1,\L+\ISI-\k}^T]^T, \nonumber
\end{align} 
where $\k$ can take values $\k \in \{-m,-,m+1,\cdots, m+\ISI\}$. We rewrite $\G(\ati)$ as 
\bea
\G(\ati)  &\define&  \Mat{H}\diag(\ati) \Mat{E} =  \ba{c} \gm_{-m}^T \\ \gm_{-m+1}^T \\ \vdots \\ \gm_{m+\ISI}^T   \ea \Mat{E}, \label{eqn:GG}
\eea
recall the definition of $\G(\ati)$ from (\ref{eqn:G}). From the observed structure of $\mat{g}_\k$ 
it can be clearly seen from (\ref{eqn:GG}) that $\G(\ati)$ is a \emph{sparse matrix} with many zero entries. The matrix $\G(\ati)$ is an $(\ISI+1)$-banded matrix, see~\cite{Golub}, p. 16. As it is well-known in the literature on ISI channels, it is efficient to employ \emph{dynamic programming} techniques to solve both problems (\ref{eqn:Max}), by exploiting this $(\ISI+1)$-banded sparsity~\cite{Viterbi}.

\newcommand{\C}{\mathcal{C}}
\newcommand{\bet}{\beta}
\renewcommand{\tt}{\tau}
\renewcommand{\ss}{\bar{s}}
\newcommand{\smat}{\bar{\s}}

\begin{algorithm}[!t]
	\SetLine
	\linesnumbered
	\nocaptionofalgo
	\SetKwInput{Init}{Initialize}  \SetKwInput{Input}{Input} \SetKwInput{Output}{Output} \SetKwInput{Conv}{\emph{Convention}}  \SetKwInput{ConvT}{\mbox{ \hspace{40pt}  }} 
	\Conv{Set $\C_0 := -\infty$ and also set values $\C_j:=0 $ for all $|j| > \L$;}
	
	\ConvT{Denote the length-$\ISI$ binary vector by $\smat\define [\ss_{\ISI-1}, \ss_{\ISI-2},\cdots, \ss_0]^T$;}
	\Input{Matrix $\G(\ati)$; Vector of constants $\pmb{\C} = [\C_{-\L}, \C_{-\L+1}, \cdots, \C_{-1}, \C_1, \C_2, \cdots, \C_\L]^T$;}
	\Output{Value stored in $\bet_{\L+\ISI} (\smat) = \bet_{\L+\ISI}(\mat{0})$;
	}
	
	\Init{For all $\smat \in \{0,1\}^{\ISI} $, set the values \[ 
		\renewcommand{\arraystretch}{.7}
		\bet_{-\L-1} (\smat)  := \left\{  
		\begin{array}{cl}
			0 & \mbox{ if } \smat = \mat{0} ,\\
			-\infty & \mbox{ otherwise } .
		\end{array}
		\right.
	\]}
	\ForAll{$\tt \in \{-\L, -\L + 1, \cdots, \L+ \ISI$ \}}{
 		\ForAll{$\smat \in \{0,1\}^{\ISI} $}{
			Set the value $\alp = \alp(\smat) := \sum_{j=0}^{\ISI-1} h_j a_{\tt-j} \ss_{j}$. Set the states $\smat_0$ and $\smat_1$ as 
			\mbox{}~~~~~~~~
$\begin{array}{cc}
				\smat_0 &:= [0, \ss_{\ISI-1}, \cdots, \ss_2, \ss_1]^T, \\
				\smat_1 &:= [1, \ss_{\ISI-1}, \cdots, \ss_2, \ss_1]^T.
			\end{array}$\;
			Compute $\bet_\tt (\smat) : = \max \{- \alp^2 + \bet_{\tt-1} (\smat_0), \C_{\tt-\ISI} - [h_\ISI a_{\tt-\ISI} + \alp]^2 + \bet_{\tt-1} (\smat_1)\}$\;
		}
	}
	
	\caption{\textbf{Procedure 2}: Solve $\displaystyle \max_{\mathbf{s} \in \{0,1\}^{2m}} \mathbf{s}^T\pmb{\C} - |\mathbf{G}(\mathbf{a})\mathbf{s}|^2$ using Dynamic Programming}
	\label{proce:DP}
\end{algorithm}

It is clear 
that the inner product $\gm_\k^T \mat{e}_j$ extracts the $j$-th component of the vector $\gm_\k^T $, i.e., 
\bea
	\gm_\k^T \mat{e}_{\k-j} &=&  \left\{
	\begin{array}{cc}
		h_j \cdot a_{\k-j}  & \mbox{ if } 0 \leq j \leq \ell,  \\
		0 &  \mbox{ otherwise },
	\end{array}\right. \label{eqn:GGG}
\eea
where $j$ satisfies $|j| \leq m + \ISI$. Both problems (\ref{eqn:Max}) are optimized over all $\s \in \{0,1\}^{2m}$; we index
\[
 \s\define [s_{-\L}, s_{-\L+1}, \cdots, s_{-1}, s_1, s_2, \cdots , s_\L]^T. 
\]
It is clear that by using (\ref{eqn:GGG}), the following is true for all vectors $\gm_\k^T$, 
\bea
\gm_\k^T  \Mat{E} \s &=&  \sum_{j=-\L-\ISI}^{\L+\ISI} (\gm^T_\k\mat{e}_j) \cdot s_j \nn &=&  \sum_{j=0}^\ISI h_j \cdot a_{\k-j} \cdot s_{\k-j}, \label{eqn:G4}
\eea
if we set $s_0 =0$ and $s_\k =0$ for all $|\k| > \L$.

Define the length-$(2m)$ vector $\pmb{\C} \define [\C_{-\L}, \C_{-\L+1}, \cdots, \linebreak[1]\C_{-1}, \C_1, \C_2, \cdots, \C_\L]^T$. Set $\C_0 := -\infty$ and $\C_\k :=0$ for all $|\k| > m$. By setting 
\bea
\pmb{\C} &:=& \Q_i\Del\mat{u} + [\G(\ati)]^T \cdot \Mat{T} (\1-\ati) \nonumber
\eea
and 
\bea
\pmb{\C} &:=& \Q_i\Del\mat{u} + [\G(\ati)]^T \cdot [\Mat{T} (\1-\ati) - 2a_0 \cdot \mat{h}_0 ], \nonumber
\eea
respectively, we can solve both problems (\ref{eqn:Max}) as
\begin{align}
	\max&_{\s \in \{0,1\}^{2\L}} \; \s^T\pmb{\C} - |\G(\ati)\s|^2 \nn 
	&= \max_{\s \in \{0,1\}^{2\L}}  \sum_{\k=-\L}^{\L+\ISI} \C_\k \cdot s_\k - (\gm^T_\tt \Mat{E} \s)^2,\label{eqn:Max2}
\end{align}
where the $\tt$-th term $\gm^T_\tt \Mat{E} \s = \sum_{j=0}^\ISI h_j  a_{\k-j}  s_{\k-j}$.
For the sake of completeness, we provide the dynamic programming procedure that solves (\ref{eqn:Max2}). The \emph{dynamic programming  state} at time $\tau$ equals the length-$\ell$ vector of binary symbols $[s_{\tau-\ell+1},s_{\tau-\ell+2},\cdots, s_{\tau}]^T \in \{0,1\}^{\ell}$. For the benefit of readers knowledgeable in dynamic programming techniques, we illustrate the time evolution of the dynamic programming states in Figure \ref{fig:CS2}. Dynamic programs can be solved with a complexity that is \emph{linear} in the state size~\cite{Viterbi}; in our case we have $2^\ell$ states. The dynamic programming procedure optimizing (\ref{eqn:Max2}) is given in Procedure \ref{proce:DP}.

\newcommand{\Syms}{\bar{\Sym}}

The second part of this subsection addresses the following separate issue. 
Consider the case where some of the probabilities $\Pr{\randb{A}_{t}= \ati}$ equal $0$; one example of such a case is where a modulation code is present in the system~\cite{RLL,Immink}. In these cases we consider the subset $ \Syms \subset \Sym $, explicitly written as 
\bea
	\Syms = \Syms_{t} \define \left\{ \mat{a} \in \Sym :  \Pr{ \bigcap_{j=-m}^m  \{ A_{t+j}=a_{j} \}} =0 \right\} \label{eqn:excl}
\eea
for each time instant $t$. 

If we consider the subsets $ \Syms \subset \Sym $, then Procedure \ref{proce:pdfXmY} has to be modified. The modification of Procedure \ref{proce:pdfXmY} is given as Procedure \ref{proce:FT}; this modification will be justified in the upcoming Section \ref{sect:dens}). Note that Line 4 of Procedure \ref{proce:FT} may also be efficiently solved using dynamic programming techniques.

Thus far, we have completed the statement of our main result Theorem \ref{thm:1} and the two main Corollaries 1 and 2. We have given Procedures \ref{proce:pdfXmY}-\ref{proce:FT} (also see Appendix \ref{app:spec}), used to efficiently evaluate the given closed-form expressions.

\begin{algorithm}[!t]
	\SetLine
	\linesnumbered
	\nocaptionofalgo
	\SetKwInput{Init}{Initialize} 
	\Init{Set $F_{\Xtk-\Ytk}(\mat{r}) := 0$ for all $\mat{r} \in \Real^\kappa$;}
	\While{$F_{\Xtk-\Ytk}(\mat{r})$ not converged}{
	  Perform Lines 2-3 of Procedure \ref{proce:pdfXmY}\;
	  Compute $\d_i = \d_i(\om, \atk)$ for all $i \in \{1,2,\cdots, n\}$ by computing
	  \bea
		\mathop{\max}_{k : \; \Alp(\Mat{E}\s_k,  \mat{a}) \in \Syms_{t_i}} \s_k^T\Q_i \Del \mat{u} + \mu_k(\ati), \nn
		\mathop{\max}_{k : \; \Alp(\Mat{E}\s_k + \mat{e}_0, \mat{a} )\in \Syms_{t_i}} \s_k^T\Q_i \Del \mat{u} + \nu_k(\ati), \nonumber
	 \eea
	 see (\ref{eqn:m}), where $\mu_k(\ati)$ and $\nu_k(\ati)$ denote the $k$-th components of $\muY(\ati)$ and $\nuX(\ati)$, see (\ref{eqn:muY}) and (\ref{eqn:nuX}). 
	Both $\Mat{E}$ and $\mat{e}_0$ are given in (17) and (18). 
	Also, the vector $\Alp(\mat{e}, \ati) = [\alp_{-\L-\ISI},$ $ \alp_{-\L-(\ISI-1)}, $$  \cdots,  \alp_{\L+\ISI} ]^T$ satisfies \newline \mbox{}~~~~~~~~
	  $
		\alp_j =$$ \alp_j(e_j,a_j)  = \left\{\begin{array}{rl}
			-a_j & \mbox{ if } e_j = 1,  \\
			a_j & \mbox{ if } e_j = 0.
				\end{array} \right.
		$\;
		Perform Lines 5-6 of Procedure \ref{proce:pdfXmY}\;
	  }	
	\caption{\textbf{Procedure 3}: Evaluating $F_{\pmb{\rand{X}}_{\textbf{t}_1^n}-\pmb{\rand{Y}}_{\textbf{t}_1^n}}(\textbf{r})$, for candidate subsets $\Syms \subset \Sym$, see (\ref{eqn:excl})}
	\label{proce:FT}
\end{algorithm}

\renewcommand{\Beta}{\pmb{\Gam}}
\renewcommand{\alp}{\theta}
\renewcommand{\mat}[1]{\begin{bf} #1 \end{bf}}

\section{Proof of Theorem 1 and Some Comments} \label{sect:dens}

\subsection{Proof of Theorem \ref{thm:1}} \label{ssect:prf}

We begin by showing the correctness of Theorem \ref{thm:1}, which was stated in the previous section. Define the random variable 
\bea
\v_t &\define& \rand{A}_{t} \cdot \hO^T\pmb{\rand{W}}_t. \label{eqn:v} 
\eea
It is easy to verify that $\v_t$ is Gaussian: recall that $\pmb{\rand{W}}_t \define [\rand{W}_{t-M},\rand{W}_{t-M+1},\cdots,\rand{W}_{t+M+I}]^T$ is the neighborhood of (Gaussian) noise samples. To improve clarity, we shall introduce the following new notation, both used only in this section
\bea
 \!\!\! \!\!\! \!\!\! \alp(\randb{A}_t) \!\!\! &\define& \!\!\! \rand{A}_t \cdot [\Mat{T}(\1 - \randb{A}_t)]^T \mat{h}_0 - |\mat{h}_0|^2, \nn
\!\!\! \!\!\! \!\!\!  \Beta = \Beta(\Atk) \!\!\! &\define& \!\!\!  \bdiag\left(\GMats\right).
  \label{eqn:short}
\eea
Using (\ref{eqn:short}), we may now more compactly write 
\bea
\Mat{Q}\Del^2\Mat{Q}^T &=& \Beta^T \Kw \Beta, \nn
\Mat{F}(\Atk) &=& \diag(A_{t_1},A_{t_2},\cdots, A_{t_n}) \otimes \mat{h}_0^T  \Kw \Beta \nn &&\cdot [\I_n \otimes \S\S^T] \cdot 
\Q \Del\Iv, \nn
\pmb{\eta}(\U,\Atk) &=& [\alp(\randb{A}_{t_1}),\alp(\randb{A}_{t_2}),\cdots, \alp(\randb{A}_{t_n})]^T + \Mat{F}(\Atk)\U, \nn
\eea
where $\Q=\Q(\Atk)$ and $\Del=\Del(\Atk)$ are given in Definition \ref{defn:QD}, matrix $\Mat{F}(\Atk)$ in (\ref{eqn:F}), and $\pmb{\eta}(\U,\Atk)$ in (\ref{eqn:eta}).

\begin{pro} \label{pro:XaYeqdist}
The random variables $\rand{X}_t$ and $\rand{Y}_t$ in (\ref{eqn:XaY}) can be written as 
\bea
\rand{X}_t &=& \max \left( [\G(\randb{A}_t)\S]^T \pmb{\rand{W}}_t + \nuX(\randb{A}_t)  +   [\v_t  + \alp(\randb{A}_t)]\cdot \1 \right),\nn
\rand{Y}_t &=& \max \left( [\G(\randb{A}_t)\S]^T \pmb{\rand{W}}_t + \muY(\randb{A}_t) \right),\nonumber
\eea
where $\alp(\randb{A}_t) \define \rand{A}_t \cdot [\Mat{T}(\1 - \randb{A}_t)]^T \mat{h}_0 - |\mat{h}_0|^2$ as given in (\ref{eqn:short}).  \hspace*{\fill}\IEEEQEDopen
\begin{proof} \rm
We expand $\D(\pmb{\rand{A}}_t,\vec{a})$ in (\ref{eqn:D}) by substituting for $\pmb{\rand{Z}}_t$ using (\ref{eqn:isi_chan}) to get 
\begin{align} 
&\D(\pmb{\rand{A}}_t,\vec{a}) \nn
&= |\pmb{\rand{Z}}_t - \Mat{T}\1 - \Mat{H}\pmb{\rand{A}}_t |^2 - |\pmb{\rand{Z}}_t - \Mat{T}\1 - \Mat{H}\vec{a}|^2 \nn
&= |-\pmb{\rand{W}}_t + \Mat{T}(\pmb{\rand{A}}_t - \1)|^2  \nn
&~~~~- |-\pmb{\rand{W}}_t + \Mat{T}(\pmb{\rand{A}}_t - \1) + \Mat{H}(\pmb{\rand{A}}_t - \vec{a})|^2 \nn
&= -2[-\pmb{\rand{W}}_t+ \Mat{T}(\pmb{\rand{A}}_t - \1)]^T\Mat{H}(\pmb{\rand{A}}_t - \vec{a}) - |\Mat{H}(\pmb{\rand{A}}_t - \vec{a})|^2. \nn\label{xmas}
\end{align}
We substitute (\ref{xmas}) into the definition of $\rand{X}_t$ and $\rand{Y}_t$ in (\ref{eqn:XaY}) to obtain
\bea 
\rand{X}_t &=& \mathop{\max_{\mat{a} \in \Sym }}_{a_0 \neq \rand{A}_t} [\pmb{\rand{W}}_t+ \Mat{T}(\1-\randb{A}_t)]^T\left( \frac{1}{2} \cdot \Mat{H}(\randb{A}_t - \mat{a})\right)\nn &&  - \left| \frac{1}{2} \cdot \Mat{H}(\randb{A}_t - \mat{a})\right|^2 , \nn
\rand{Y}_t &=& \mathop{\max_{\mat{a} \in \Sym }}_{a_0 = \rand{A}_t} [\pmb{\rand{W}}_t+ \Mat{T}(\1-\randb{A}_t)]^T\left( \frac{1}{2} \cdot \Mat{H}(\randb{A}_t - \mat{a})\right) \nn && - \left| \frac{1}{2} \cdot \Mat{H}(\randb{A}_t - \mat{a})\right|^2. \label{eqn:XaY01}
\eea
Using (\ref{eqn:mate}) and Definitions (\ref{eqn:Sym}), (\ref{eqn:Emat}) and (\ref{eqn:S}), we establish the following equality of sets
\begin{align}
 &\Ev{ \frac{1}{2} (\rand{A}_t - \mat{a}) : \mat{a} \in \Sym, a_0 \neq \rand{A}_t } \nn &\quad\quad\quad= \Ev{\diag(\randb{A}_t) \Mat{E} \s_j + A_t \cdot \mat{e}_0 : 0 \leq j \leq 2^{2m}-1 }, \nn
 &\Ev{ \frac{1}{2} (\rand{A}_t - \mat{a}) : \mat{a} \in \Sym, a_0 =  \rand{A}_t } \nn &\quad\quad\quad= \Ev{\diag(\randb{A}_t) \Mat{E} \s_j: 0 \leq j \leq 2^{2m}-1 }. \label{eqn:mod}
\end{align}
Next, we utilize both (\ref{eqn:XaY01}) and (\ref{eqn:G}) to rewrite (\ref{xmas}) as 
\bea
\rand{X}_t  \!\!&=& \!\!\!\!\!\!\! \max_{j \in \{0,1,\cdots,2^{2\L}-1\}} [\pmb{\rand{W}}_t+ \Mat{T}(\1-\randb{A}_t)]^T [\Mat{G}(\randb{A}_t)\s_j + A_t \mat{h}_0]  \nn && - |\Mat{G}(\randb{A}_t)\s_j + A_t \mat{h}_0|^2, \nn
\rand{Y}_t\!\! &=& \!\!\!\!\!\!\!\max_{j \in \{0,1,\cdots,2^{2\L}-1\}} [\pmb{\rand{W}}_t+ \Mat{T}(\1-\randb{A}_t)]^T [\Mat{G}(\randb{A}_t)\s_j ] \nn && - |\Mat{G}(\randb{A}_t)\s_j |^2. \label{eqn:XaY02}
\eea
By the definition of $\muY(\randb{A}_t)$ in (\ref{eqn:muY}) and $\Mat{S}$ in Definition (\ref{eqn:S}), the expression for $Y_t$ in the proposition statement follows from (\ref{eqn:XaY02}). For $X_t$, we continue to expand (\ref{eqn:XaY02}) to get
\begin{align}
&\rand{X}_t \nn &= \max \Big( [\Mat{G}(\randb{A}_t)\Mat{S} ]^T\randb{W}_t  + \overbrace{\muY(\randb{A}_t) - 2A_t \cdot \mat{h}_0^T \Mat{G}(\randb{A}_t)\S}^{\nuX(\randb{A}_t)} 
  \nn
& \quad \quad \quad + \underbrace{A_t \cdot \mat{h}_0^T \randb{W}_t}_{\v_t} \cdot \1  +  \{\underbrace{A_t [\Mat{T}(\1-\randb{A}_t)]^T\mat{h}_0 - |\mat{h}_0|^2}_{\alp(\randb{A}_t)} \} \cdot  \1	\Big),
\nonumber
\end{align}
in the same form as in the proposition statement, where $\nuX(\randb{A}_t)$ is defined in (\ref{eqn:nuX}), and $\v_t$ in (\ref{eqn:v}), and $\alp(\randb{A}_t)$ in (\ref{eqn:short}).
\end{proof}
\end{pro}

Recall $\Q = \Q(\Atk)$ and $\Del = \Del(\Atk)$ from Definition \ref{defn:QD}. To prove Theorem \ref{thm:1} we require the following lemma. 

\renewcommand{\zeta}{U}

\begin{lem} \label{lem:eig}
Let $\pmb{\zeta}$ denote a standard zero-mean identity-covariance Gaussian random vector of length-$(2\L\kappa)$. Recall $\Wtk$ in (\ref{eqn:Wtk}). The following transformation of random vectors holds
\begin{align}
  &\ba{c} \S^T \Q_1(\Atk) \\ \S^T \Q_2(\Atk)  \\ \vdots \\ \S^T \Q_n (\Atk) \ea \Del(\Atk) 
\pmb{\zeta} \nn
&= \ba{@{}c@{}c@{}c@{}c@{}} \Mat{G}(\randb{A}_{t_1})\S \\ & \Mat{G}(\randb{A}_{t_2})\S  \\ & & \ddots \\ & & & \Mat{G}(\randb{A}_{t_n})\S \ea^T 
\ba{c} \pmb{\rand{W}}_{t_1}  \\ \pmb{\rand{W}}_{t_2} \\ \vdots  \\ \pmb{\rand{W}}_{t_n} 
\ea,
\end{align}
or more concisely we equivalently write
\begin{align}
	(\I_n \otimes \S^T)&\cdot \Q(\Atk)\Del(\Atk)\pmb{\zeta} \nn & = (\I_n \otimes \S^T)  \cdot \Beta(\Atk)^T \Wtk.  \label{eqn:eig}
\end{align} 
using $\Q(\Atk)$ in (\ref{eqn:Qpart}) and $\Beta(\Atk)$ in (\ref{eqn:short}). \hspace*{\fill}\IEEEQEDopen
\begin{proof} \rm
After conditioning on $\Atk$, both vectors that appear on either side of (\ref{eqn:eig}), are seen to be zero mean Gaussian random vectors (recall that $W_t$ is zero mean). Therefore to prove the lemma, we only need to verify that after conditioned on $\Atk$, both l.h.s. and r.h.s. of (\ref{eqn:eig}) have the same covariance matrix. This is easily done by using property i) of $\Q = \Q(\Atk)$ in Definition \ref{defn:QD}, which yields
\bea
	\E\left\{\Q\Del\pmb{\zeta}\pmb{\zeta}^T\Del\Q^T \right | \Atk\}  &=& \Q(\Atk)\Del(\Atk)^2\Q(\Atk)^T \nn &=& \Beta(\Atk)^T\Kw\Beta(\Atk). \nonumber
\eea
\end{proof}
\end{lem}

We are now ready to prove Theorem \ref{thm:1}. The proof is split up into the following two seperate cases : 
\bitm
	
	\item $\rank[\Beta(\Atk)^T \Kw \Beta(\Atk)]=2mn$, and
	\item $\rank[\Beta(\Atk)^T \Kw \Beta(\Atk)] < 2mn$ for some realization $\Atk=\atk$.
\eitm
We begin with the first case.

\renewcommand{\Omega}{U}

\begin{proof}[Proof of Theorem \ref{thm:1} when $\rank(\Beta(\Atk)^T\pmb{\mathcal{K}}_{\pmb{\rand{W}}}\Beta(\Atk)) = 2\L\kappa$]  \newline\indent
We first derive the following equalities 
\begin{align}
 (\Del\Iv\Q^T) &(\I_n \otimes \S \S^T) \Beta(\Atk)^T\Wtk \nn &= (\Del\Iv\Q^T) (\I_n \otimes \S \S^T)  \Q\Del\pmb{\zeta} \nn &= \Del\Iv \Del\pmb{\zeta} = \U. \label{eqn:Vinv}
\end{align}
The first two equalities follow by respectively applying properties i) and ii) of the matrix $\Q = \Q(\Atk)$. The last equality holds because by virtue of the assumption $\rank(\Beta(\Atk)^T\Kw\Beta(\Atk)) = 2mn$, in which then $\Del\Iv$ is strictly an inverse of $\Del$. Recall both $\v_{t_i} \define A_{t_i} \cdot \mat{h}_0^T \randb{W}_{t_i}$ and $\Vtk \define [\v_{t_1},\v_{t_2},\cdots, \v_{t_n}]^T$. Taking (\ref{eqn:Vinv}) together with (\ref{eqn:v}), we have the following transformation 
\bea
\ba{c}  \Vtk \\ \pmb{\zeta} \ea
= \ba{c} \diag(A_{t_1},A_{t_2},\cdots, A_{t_n}) \otimes \hO^T \\ (\Del\Iv\Q^T) (\I_n \otimes \S \S^T) \Beta(\Atk)^T \ea
\Wtk. \label{eqn:main01}
\eea
Consider the conditional event
\bea
\Ev{\Xtk - \Ytk \leq \mat{r} | \Atk , \pmb{\zeta }} \label{eqn:mainevt}
\eea
where $\mat{r}=[r_1,r_2,\cdots, r_n]^T \in \Real^{n}$. It is clear from both Proposition \ref{pro:XaYeqdist} and (\ref{eqn:main01}), that after conditioning on both $ \Atk $ and $ \pmb{\zeta }$ in (\ref{eqn:mainevt}), the only quantity that remains random in (\ref{eqn:mainevt}) is the Gaussian vector $\Vtk$. 
Using Lemma \ref{lem:eig}, we have the transformation 
\[
\S^T \Q_i(\Atk) \Del(\Atk) \pmb{\Omega} = [\G(\randb{A}_{t_i})\S]^T \randb{W}_{t_i},
\]
therefore we may rewrite both $X_{t_i}$ and $Y_{t_i}$ from Proposition \ref{pro:XaYeqdist} as
\bea
\rand{X}_{t_i} &=& \max \left(\S^T \Q_i\Del \pmb{\Omega} + \nuX(\randb{A}_{t_i}) \right) +   \v_{t_i}  + \alp(\randb{A}_{t_i}),\nn
\rand{Y}_{t_i} &=& \max \left( \S^T\Q_i\Del \pmb{\Omega} + \muY(\randb{A}_{t_i}) \right).\label{eqn:mod3}
\eea
The event (\ref{eqn:mainevt}) can then be written as
\begin{align}
 & \Ev{\Xtk - \Ytk \leq \mat{r} | \Atk , \pmb{\zeta }} 
= \bigcap_{1\leq i \leq \kappa } \Ev{X_{t_i} \leq r_i + Y_{t_i} | \Atk , \pmb{\zeta }} \nn
\!\!\!\! &= \!\!\!\! 
\bigcap_{1\leq i \leq \kappa }\left\{  \left.
\max \left(\begin{array}{l}
[\S^T \Q_i \Del\pmb{\Omega} + \nuX(\randb{A}_{t_i}) ] \\
 + \v_{t_i} + \alp(\randb{A}_{t_i})  
\end{array} \right)
\leq r_i + Y_{t_i}
\right|   \!\!\!\!
\begin{array}{c}
\Atk ,  \pmb{\zeta}\end{array} \!\!\!\!
\right\} \nn  
\!\! &= \! \!\!\!\!
\bigcap_{1\leq i \leq \kappa }  \!\!\!\left.\left\{ \!\!
\begin{array}{l}
 \v_{t_i}  + \\ \alp(\randb{A}_{t_i}) 
\end{array} \!\!\!\!
\leq \!\left(  \!\!\!\!
\begin{array}{c}
 r_i + \max \left[\S^T \Q_i\Del\pmb{\Omega} + \muY(\randb{A}_{t_i})\right]  \\
 - \max \left[\S^T\Q_i\Del\pmb{\Omega} + \nuX(\randb{A}_{t_i})\right]  
\end{array}  \!\!\! \right)
\right|    \!\!\!\!
\begin{array}{c} 
\Atk ,  \pmb{\zeta}\end{array}  \!\!\!\!
\right\}\!.\nn 
\label{eqn:main1}
\end{align}
Continuing from (\ref{eqn:main1}), we utilize (\ref{eqn:m}) to rewrite 
\bea
&& \Ev{\Xtk-\Ytk \leq \mat{r} | \Atk, \pmb{\Omega}} \nn
&=&
\bigcap_{1\leq i \leq \kappa }\left\{  \v_{t_i} + \alp(\randb{A}_{t_i }) \leq r_i + \delta_i(\pmb{\Omega},\Atk) |\Atk , \pmb{\Omega} \right\}.~~~~~ \label{eqn:main001}
\eea
We now determine both the mean and variance of $\Vtk$, after conditioning on both $\Atk$ and $\pmb{\Omega}$. 
From (\ref{eqn:main01}), we derive the formula 
\bea
\E\{\Vtk \pmb{\zeta}^T | \Atk \} &=& \diag(A_{t_1},A_{t_2},\cdots, A_{t_n}) \otimes \mat{h}_0^T \Kw  \nn 
&&\cdot~\Beta(\Atk) (\I_n\otimes\S\S^T)\Q \Del\Iv  \nn
&\define& \Mat{F}(\Atk), \label{eqn:F2}
\eea
where $\Mat{F}(\Atk)$ is given in (\ref{eqn:F}) . Next, we compute the conditional mean 
\bea
 \E \left\{\Vtk | \Atk,\pmb{\zeta}\right\} &=& \E\{\Vtk |\Atk \} + \E\{\Vtk \pmb{\zeta}^T | \Atk \}\U \nn 
 	&=&  \Mat{F}(\Atk) \U,\label{eqn:m2}
\eea
where the second equality follows from $\E\{\Vtk |\Atk \}=0$ (because $\Wtk$ has zero mean, see (\ref{eqn:v})), and substituting (\ref{eqn:F2}). The conditional covariance matrix $\mbox{\rm \Var} \left\{\Vtk  | \Atk, \pmb{\zeta} \right\} $ is obtained as follows
\begin{align}
 & \mbox{\rm \Var} \left\{\Vtk  | \Atk, \pmb{\zeta} \right\} \nn
&= \E\{\Vtk \Vtk^T| \Atk\} - \E\{\Vtk \pmb{\zeta}^T| \Atk\} \cdot \E\{\pmb{\zeta} \Vtk^T | \Atk\}   \nn
&=  \diag(A_{t_1},A_{t_2},\cdots, A_{t_n}) \otimes \mat{h}_0^T \Kw \nn &~~~\cdot\diag(A_{t_1},A_{t_2},\cdots, A_{t_n}) \otimes \mat{h}_0 
  - \Mat{F}(\Atk)\Mat{F}(\Atk)^T \nn &\define \Kv(\Atk) ,\label{eqn:nu2}
\end{align} 
where $\Kv(\Atk)$ is given in (\ref{eqn:nu}). The expression for $F_{\Xtk-\Ytk} (\mat{r})$ in Theorem \ref{thm:1} now follows easily from (\ref{eqn:main001}) 
\begin{align}
\big\{\Xtk-\Ytk \leq &\mat{r} | \Atk, \pmb{\Omega} \big\}\nn  &=  
 \Big\{\Vtk +  [\alp(\randb{A}_{t_1}),\alp(\randb{A}_{t_2}),\cdots, \alp(\randb{A}_{t_n})]^T \nn
 & \quad\quad\; \leq \mat{r} +   \dmat(\U,\Atk) \Big|\Atk, \U \Big\} \nonumber
\end{align}
and noticing that the random vector 
\bea
	\Vtk +  [\alp(\randb{A}_{t_1}),\alp(\randb{A}_{t_2}),\cdots, \alp(\randb{A}_{t_n})]^T
\eea
is (conditionally on $\Atk$ and $\U$) Gaussian distributed with distribution function 
	\[ \Phi_{\Kv(\Atk)}(\mat{r} - \pmb{\eta}(\U,\Atk)), \]
where both the conditional mean and covariance $\pmb{\eta}(\U,\Atk)$ and $\Kv(\Atk)$, are given respectively in (\ref{eqn:m2}) and (\ref{eqn:nu2}). 
\end{proof}

Next we consider the other case where the rank of $\Beta(\Atk)^T\textbf{K}_{\pmb{\rand{W}}}\Beta(\Atk) < 2mn$ for some value of $\Atk = \atk$.  In this case, the arguments of the preceding proof fail in equation (\ref{eqn:Vinv}), where the final equality does not hold because then $\Del\Iv$ is strictly not the inverse of $\Del$. However as we soon shall see, the expression for $F_{\Xtk-\Ytk}(\mat{r})$ in Theorem \ref{thm:1} still holds for this case.

\renewcommand{\Om}{\pmb{\Omega}_1^j}
\newcommand{\Qb}{\bar{\Q}}
\newcommand{\Delb}{\bar{\Del}}
\renewcommand{\O}{\pmb{\Omega}}

\begin{proof}[Proof of Theorem \ref{thm:1} when $\rank(\Beta(\Atk)^T\pmb{\mathcal{K}}_{\pmb{\rand{W}}}\Beta(\Atk)) < 2\L\kappa$ for some $\Atk= \atk$] \rm \newline\indent Recall that the matrix $[\Del(\Atk)]\Iv = \Del\Iv$ is formed by only reciprocating the non-zero diagonal elements of $\Del(\Atk) = \Del$. 
For a particular realization $\Atk=\atk$, let the value $j=\rank(\Beta(\Atk)^T\Kw\Beta(\Atk) )$ equal the rank of the matrix $\Beta(\Atk)^T\Kw\Beta(\Atk) $. Consider what happens if $j < 2mn$. Without loss of generality, assume that all non-zero diagonal elements of $\Del(\Atk) = \Del$, are located at the first $j < 2mn$ diagonal elements of $\Del$. Define the following size-$j$ quantities
\bitm
\item the random vector $\Om= [\Omega_1,\Omega_2, \cdots, \Omega_j]^T$, a truncated version of $\U= [\Omega_1,\Omega_2, \cdots, \Omega_{2mn}]^T$. 
\item the size $2mn$ by $j$ matrix $\Qb$, containing the first $j$ columns of the $\Q$, see Definition \ref{defn:QD}. 
\item the size $j$ diagonal square matrix $\Delb$, containing the $j$ positive diagonal elements of $\Del$, also see Definition \ref{defn:QD}. 
\eitm
If we substitute the new quantities $\Om$, $\Qb$ and $\Delb$ for $\U$, $\Q$ and $\Del$ in equation (\ref{eqn:Vinv}), it is clear that (\ref{eqn:Vinv}) holds true, i.e.,
\begin{align}
 (\Delb\Iv\Qb^T)& (\I_n \otimes \S\S^T ) \Beta(\Atk)^T\Wtk \!\!\!\! \nn &=  \!(\Delb\Iv\Qb^T) (\I_n \otimes \S\S^T )  \Qb\Delb\Om \nn &= \Delb\Iv \Delb \Om = \Om,
\end{align}
where note from Definition \ref{defn:QD} that it must be true that $\Qb^T(\I_n \otimes \S\S^T ) \Qb = \I_j$, here $\I_j$ is the size $j$ identity matrix. Hence, Theorem \ref{thm:1} clearly holds when we substitute $\Om$, $\Qb$ and $\Delb$ for $\U$, $\Q$ and $\Del$.

Further, we can verify the following facts:
\bitm
\item $\Qb_i \Delb_i \Om = \Q_i\Del \pmb{\Omega}$, and therefore
\item $\dmat(\Om,\Atk)= \dmat(\pmb{\Omega},\Atk)$. Also, 
\item $\Mat{F}(\Atk)$ remains unaltered whether we use $\Q,\Del$ or $\Qb, \Delb$, therefore 
\item $\pmb{\eta}(\Om,\Atk)= \pmb{\eta}(\pmb{\Omega},\Atk)$. Also, 
\item $\Kv(\Atk)$ remains unaltered whether we use $\Q,\Del$ or $\Qb, \Delb$.
\eitm 
Thus we conclude that 
\begin{align}
&\E \Ev{ \left.\Phi_{\Kv(\Atk)}(\dmat(\Om,\Atk) - \pmb{\eta}(\Om,\Atk)) \right|\Atk }\nn & = \E \Ev{ \left. \Phi_{\Kv(\Atk)}(\dmat(\O,\Atk) - \pmb{\eta}(\O,\Atk)) \right| \Atk }\nonumber 
\end{align}
must hold, and thus Theorem \ref{thm:1} must be true even when $\rank [\Beta(\Atk)^T\Kw\Beta(\Atk)] < 2mn$ for certain values of $\Atk = \atk$. 
\end{proof}

We have thus far completed our proof of Theorem \ref{thm:1}; we next show an upper bound for the rank of the matrix $\Kv(\Atk)$ in (\ref{eqn:nu2}). We point out that $\Kv(\Atk)$ sometimes may even have rank $0$, i.e. $\Kv(\Atk)$ equals the zero matrix.

\subsection{Other comments}

\renewcommand{\Beta}{\pmb{\beta}}
\renewcommand{\alp}{\alpha}

The following proposition states that the rank of $\Kv(\Atk)$ depends on both the chosen time instants $\{t_1,t_2,\cdots, t_\kappa\}$, and the MLM truncation length $\L$. The following proposition gives the upper bound on $\rank(\Kv(\Atk))$.

\begin{pro} \label{pro:rankKv}
The rank of $\Kv(\Atk)$ equals at most the number of time instants $t \in \{t_1,t_2,\cdots, t_\kappa\}$, that satisfy $|t - t'| > \L$ for all $t' \in \{t_1,t_2,\cdots, t_\kappa\}\setminus \{t \}$.  \hspace*{\fill}\IEEEQEDopen
\end{pro}

Proposition \ref{pro:rankKv} is proved using the following lemma. 
\newcommand{\sj}{\mat{s}}

\begin{lem} \label{lem:WandV}
If two time instants $t_1$ and $t_2$ satisfy $|t_1 - t_2| \leq \L$, then observation of $[\G (\randb{A}_{t_1}) \S ]^T\pmb{\rand{W}}_{t_1}$ uniquely determines $\rand{V}_{t_2} \define A_{t_2}\cdot \hO^T\pmb{\rand{W}}_{t_2}$ (and vice versa observation of $[\G (\randb{A}_{t_2}) \S ]^T\pmb{\rand{W}}_{t_2}$ uniquely determines $\rand{V}_{t_1} \define  A_{t_1} \cdot \hO^T\pmb{\rand{W}}_{t_1}$). \hspace*{\fill}\IEEEQEDopen
\begin{proof} \rm
Recall that $\rand{V}_{t_2}$ equals
\bea
\rand{V}_{t_2} \define   A_{t_2} \cdot \hO^T \pmb{\rand{W}}_{t_2} = A_{t_2} \cdot \left( h_0 \rand{W}_{t_2} + \cdots + h_I \rand{W}_{t_2+I}. \right)\nonumber
\eea
If the condition $|t_1 - t_2| \leq \L$ is satisfied, then $\rand{W}_{t_2}, \cdots, \rand{W}_{t_2+I}$ is a length-$(I+1)$ subsequence of $\pmb{\rand{W}}_{t_1} \define [\rand{W}_{t_1-\L}, \rand{W}_{t_1-\L+1}, \linebreak[3]\cdots, \rand{W}_{t_1+\L+I}]^T$. From the definition of $\S$ (see (\ref{eqn:S})) and because $|t_1 - t_2| \leq \L$, then the matrix $\S$ must have a column $\sj$ that satisfies $\Mat{E}\sj = \mat{e}_{t_2-t_1}$, see (\ref{eqn:Emat}). Then for this particular column $\sj$ we have  
\bea
[\G (\randb{A}_{t_1}) \sj ]^T\pmb{\rand{W}}_{t_1} \!\!\!&=& \!\!\![\Mat{H}\diag(\randb{A}_{t_1})\Mat{E}\sj]^T\pmb{\rand{W}}_{t_1} \nn
&=& 
 A_{t_2} \cdot [\Mat{H}\mat{e}_{t_2-t_1}]^T\pmb{\rand{W}}_{t_1} \nn  &=& A_{t_2} \cdot \hO^T  \pmb{\rand{W}}_{t_2} \define \rand{V}_{t_2}, \nonumber
\eea 
where the second equality holds because $\sj$ satisfies $\diag(\randb{A}_{t_1})\Mat{E}\sj = \diag(\randb{A}_{t_1}) \mat{e}_{t_2-t_1} = A_{t_2}\cdot \mat{e}_{t_2-t_1}$, and also
\begin{align}
 &[\Mat{H}\mat{e}_{t_2-t_1}]^T\pmb{\rand{W}}_{t_1} \nn &= [\Mat{H}\mat{e}_{t_2-t_1}]^T [\rand{W}_{t_1-\L}, \rand{W}_{t_1-\L+1}, \cdots, \rand{W}_{t_1+\L+I}]
 \nn &= h_0 \rand{W}_{t_2} + h_1 \rand{W}_{t_2+1} + \cdots + h_I \rand{W}_{t_2+I}. \nonumber
\end{align}
By symmetry, the same argument holds for $[\G (\randb{A}_{t_1}) \S ]^T\pmb{\rand{W}}_{t_2}$ and $\rand{V}_{t_2} \define A_{t_1}\cdot \hO \mat{W}_{t_1}$. 
\end{proof}

\end{lem}

\begin{proof}[Proof of Proposition \ref{pro:rankKv}]  Recall from (\ref{eqn:nu2}) that $\Kv(\Atk) \define \mbox{\rm \Var}\{\Vtk |\Atk, \pmb{\zeta } \}$ is the (conditional) covariance matrix of $\Vtk$. After conditioning on $\pmb{\zeta}$, the vector $\Q_i\Del\pmb{\zeta} = [\G(\randb{A}_{t_i})\S]^T \pmb{\rand{W}}_{t_i} $ is uniquely determined, see Lemma \ref{lem:eig}. Furthermore by Lemma \ref{lem:WandV}, if  $\Q_i\Del\pmb{\zeta} = [\G(\randb{A}_{t_i})\S]^T \pmb{\rand{W}}_{t_i} $ is uniquely determined then $\rand{V}_{t_j} \define A_{t_j} \cdot \hO^T \pmb{\rand{W}}_{t_j}$ is determined whenever $|t_i - t_j | \leq \L$. Thus we conclude that the only variables $\rand{V}_{t_i}$ that may contribute to the rank of $\Kv(\Atk)$, must be those with corresponding $t_i$ that are separated from all other $\{t_1,t_2,\cdots,t_\kappa \}\setminus \{ t_i\}$ by greater than $\L$.
\end{proof}

\renewcommand{\tt}{{t_i}}

\begin{rem} \label{rem:smooth}
From the expression for $F_{\Xtk-\Ytk}(\mat{r})$ in Theorem \ref{thm:1}, the distribution function $F_{\Xtk-\Ytk}(\mat{r})$ must be left-continuous~\cite{Papoulis:Text}, if the $\rank(\Kv(\Atk))=n$.
\end{rem}

\begin{table}
	\renewcommand{\arraystretch}{.8}
	\centering
	\caption{Various ISI channels in magnetic recording~\cite{PR}} 
		\begin{tabular}[t]{|c|rrr|c|}
			\hline \multirow{2}{*}{Channel} &  \multicolumn{3}{c|}{Coefficients} & Memory \\ &  $h_0$ & $h_1$ & $h_2$ & Length $\ell$\\ \hline
			PR1 &  $1$ & $1$ & - & 1\\ \hline
			Dicode & $1$ & $-1$ & - & 1\\ \hline
			PR2 & $1$ & $2$  & $1$  & 2\\ \hline
			PR4 & $1$ & $0$  & $-1$  & 2\\ \hline
		\end{tabular}
		
		\label{tab:1}
\end{table}

\begin{figure*}[!t]
	\centering
		\includegraphics[width=.8\linewidth]{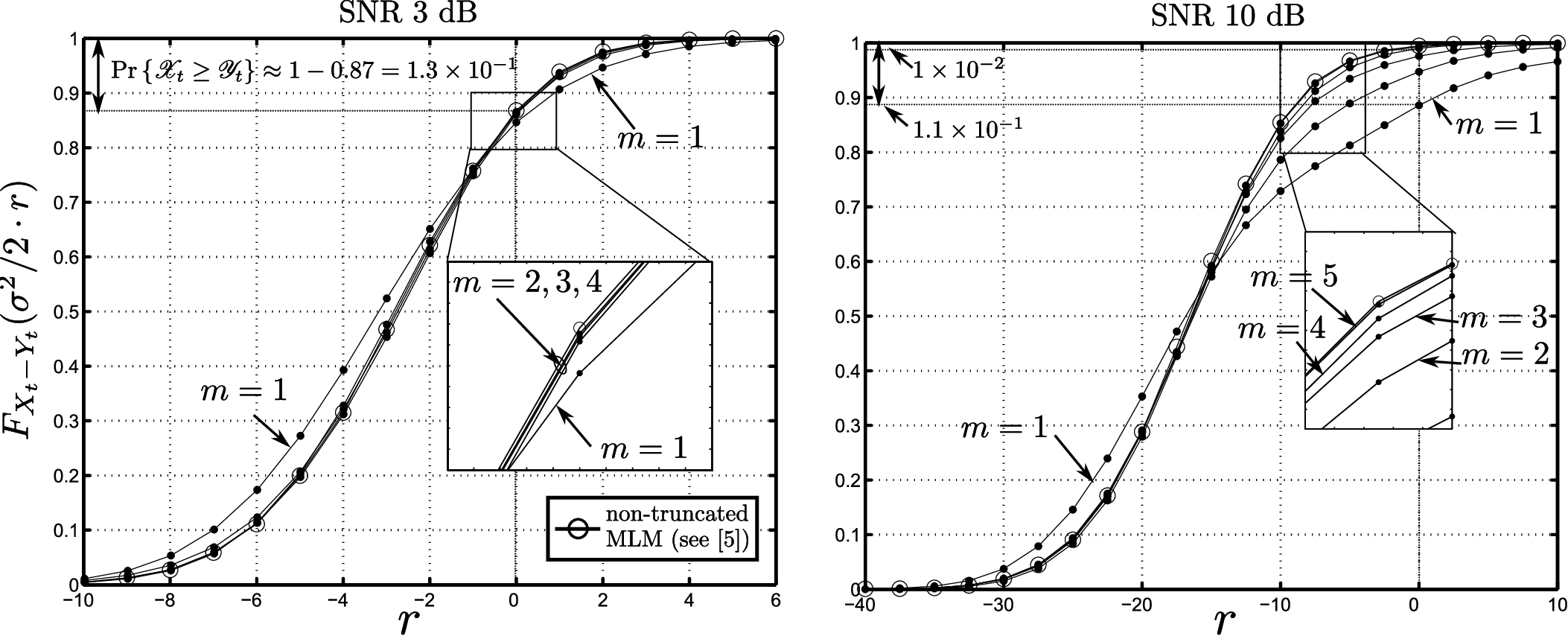}
		\caption{Marginal reliability distribution $F_{\rand{X}_t-\rand{Y}_t}(\sig^2/2 \cdot r)$ computed for the PR1 channel (see Table \ref{tab:1}). Truncation lengths $\L$ are varied from $1$ to $5$.}
	\label{fig:marg_dicode}
\end{figure*}

We conclude this section by verifying the correctness of Procedure \ref{proce:FT}, used to evaluate $F_{\Xtk-\Ytk}(\mat{r})$ when candidate subsets $\Syms \subset \Sym$ (see (\ref{eqn:excl})) are considered. The only difference between Procedures \ref{proce:pdfXmY} and \ref{proce:FT}, is that Line 3 of Procedure \ref{proce:FT} replaces Line 4 of Procedure \ref{proce:pdfXmY}. First verify that the following equality of sets is true
\begin{align}
&\Ev{\mat{a} \in \Syms_{t_i} : a_0 \neq A_{t_i}} \nn &~~~= \Ev{\Alp(\Mat{E}\s_k + \mat{e}_0, \randb{A}_{t_i})\in \Syms_{t_i} : 0 \leq k \leq 2^{2m}-1},\nn
&\Ev{\mat{a} \in \Syms_{t_i} : a_0 = A_{t_i}} \nn &~~~= \Ev{\Alp(\Mat{E}\s_k , \randb{A}_{t_i})\in \Syms_{t_i} : 0 \leq k \leq 2^{2m}-1}, \label{eqn:mod2}
\end{align}
where here the function $\Alp(\mat{e}, \randb{A}_{t_i})$ is given in Line 3 of Procedure \ref{proce:FT}. Next perform the following verifications in the order presented:
\bitm
\item Replace $\Sym$ by $\Syms_{t_i}$ in the definitions of $R_{t_i}$ in (\ref{relt}). Replace $\Sym$ by $\Syms_{t_i}$ in both $X_{t_i}$ and $Y_{t_i}$ in (\ref{eqn:XaY}). The validity of Proposition \ref{relprop} remains unaffected.
\item Replace $\Sym$ by $\Syms_{t_i}$ in the proof of Proposition \ref{pro:XaYeqdist}. The change first affects the proof starting from (\ref{eqn:XaY01}), and (\ref{eqn:mod}) needs to be slightly modified using (\ref{eqn:mod2}). The new Proposition \ref{pro:XaYeqdist} finally reads 
\begin{align}
\rand{X}_\tt &= \max_{k :\; \Alp(\Mat{E}\s_k + \mat{e}_0, \randb{A}_\tt) \in \Syms_\tt} \s_k^T [\G(\randb{A}_\tt)]^T \pmb{\rand{W}}_\tt\nn & \quad+ \nu_k(\randb{A}_\tt)  +   \v_\tt  + \theta(\randb{A}_\tt),\nn
\rand{Y}_\tt &= \max_{k :\; \Alp(\Mat{E}\s_k, \randb{A}_\tt) \in \Syms_\tt} \s_k^T [\G(\randb{A}_\tt)]^T \pmb{\rand{W}}_\tt + \mu_k(\randb{A}_\tt).\nonumber
\end{align}
\item Utilize the new Proposition \ref{pro:XaYeqdist} in the proof of Theorem \ref{thm:1}. The change first affects the proof starting from (\ref{eqn:mod3}). Proceeding from (\ref{eqn:main1})-(\ref{eqn:main001}) we arrive at the new formulas
\bea
\d_i &=& \d_i(\U,\Atk) \nn &=& \max_{k :\; \Alp(\Mat{E}\s_k + \mat{e}_0, \randb{A}_\tt) \in \Syms_\tt} \s_k^T \Q_i \Del \U + \nu_k(\randb{A}_\tt) \nn
&& - \max_{k :\; \Alp(\Mat{E}\s_k, \randb{A}_\tt) \in \Syms_\tt} \s_k^T \Q_i \Del \U + \mu_k(\randb{A}_\tt). \nonumber
\eea
This is exactly the way $\d_i$ is computed in Procedure \ref{proce:FT}, Line 3.
\eitm
This concludes our verification of Procedure \ref{proce:FT}.

\section{Numerical Computations} \label{sect:egs}
\newcommand{\dist}{F_{\Xtk-\Ytk}(\sig^2/2\cdot \mat{r})}

We now present numerical computations performed for various ISI channels. To demonstrate the generality of our results, various cases will be considered. Both i) the reliability distribution $F_{\Rtk}(\mat{r})$ and ii) the symbol error probability $\PSE$ will be graphically displayed in the following manner. Recall from Corollaries \ref{cor:main} and  \ref{cor:main2} that we have $F_{\Rtk}(\mat{r})=F_{|\Xtk-\Ytk|}(\sig^2/2\cdot \mat{r})$ (here $\sig^2$ denotes the noise variance in (\ref{eqn:snr})) and $\PSE = \Pr{\Xtk \geq  \Ytk}$. Therefore, both quantities i) and ii) will be displayed utilizing a \emph{single} graphical plot of $\dist$.

The chosen ISI channels for our tests are given in Table \ref{tab:1}; these are commonly-cited channels in the magnetic recording literature~\cite{PR,Immink}.
Define the signal-to-noise (SNR) ratio as $10\log_{10} ( \sum_{i=0}^\ell h_i^2 /\sig^2)$. The input symbol distribution $\Pr{\randb{A}_t=\ati}$ will always be uniform, i.e., $\Pr{\randb{A}_t=\ati}=2^{-2(m+\ell)-1}$ see (\ref{eqn:sig_vect}), unless stated otherwise.

\newcommand{\marg}{F_{X_t-Y_t}(\sig^2/2 \cdot r)}
\subsection{Marginal distribution when the noise is i.i.d.} \label{ssect:marg1}

First, consider the case where the noise samples $\rand{W}_t$ are i.i.d, thus $\sig^2 = \E\{\rand{W}_t^2\}$. Figure \ref{fig:marg_dicode} shows the marginal distribution $\marg$ computed for the PR1 channel (see Table \ref{tab:1}) with memory $\ISI=1$.
The distribution is shown for various truncation lengths $\L=1$ to $5$, and two different SNRs : 3 dB and 10 dB. At SNR 3 dB, we observe that with the exception of $\L=1$, all curves appear to be extremely close. At SNR 3 dB, a good choice for the truncation length $\L$ appears to be $\L=2$; the computed distribution for $\L=2$ appears close to the simulated distribution. 
At SNR 10 dB, it appears that $\L=5$ is a good choice. The probability of symbol error $\Pr{\rand{B}_t \neq \rand{A}_t}=\Pr{\rand{X}_t \geq \rand{Y}_t} = 1 - F_{X_t - Y_t}(0)$ is observed to decrease as the truncation length $\L$ increases; this is expected. At SNR 3 dB, the (error) probability $\Pr{\rand{X}_t \geq \rand{Y}_t} = 1 - F_{X_t - Y_t}(0) \approx 1.4 \times 10^{-1}$ for truncation lengths $\L > 1$. For SNR 10 dB, the (error) probability $\Pr{\rand{X}_t \geq \rand{Y}_t} $ is seen to vary significantly for both truncation lengths $\L=1$ and $5$; the probability $\Pr{\rand{X}_t \geq \rand{Y}_t} \approx 1.1 \times 10^{-1} $ and $1 \times 10^{-2}$ for $m=1$ and $5$, respectively.

For the PR1 channel and a fixed truncation length $\L = 4$, the marginal distributions $\marg$ are compared across various SNRs in Figure \ref{fig:marg_comp}. As SNR increases, the distributions $\marg$ appear to concentrate more probability mass over negative values of $X_t-Y_t$. This is intuitively expected, because as the SNR increases, the symbol error probability $\Pr{\rand{B}_t \neq \rand{A}_t}=\Pr{\rand{X}_t \geq \rand{Y}_t} = 1 - F_{X_t - Y_t}(0)$ should decrease. From Figure \ref{fig:marg_comp}, the (error) probabilities $\Pr{\rand{X}_t \geq \rand{Y}_t}$ are found to be approximately $ 1.2 \times 10^{-1}, 8 \times 10^{-2}, 3 \times 10^{-2}$, and $1 \times 10^{-2}$, respectively for SNRs 3 to 10 dB.

\begin{figure}[!t]
	\centering
		\includegraphics[width=.9\linewidth]{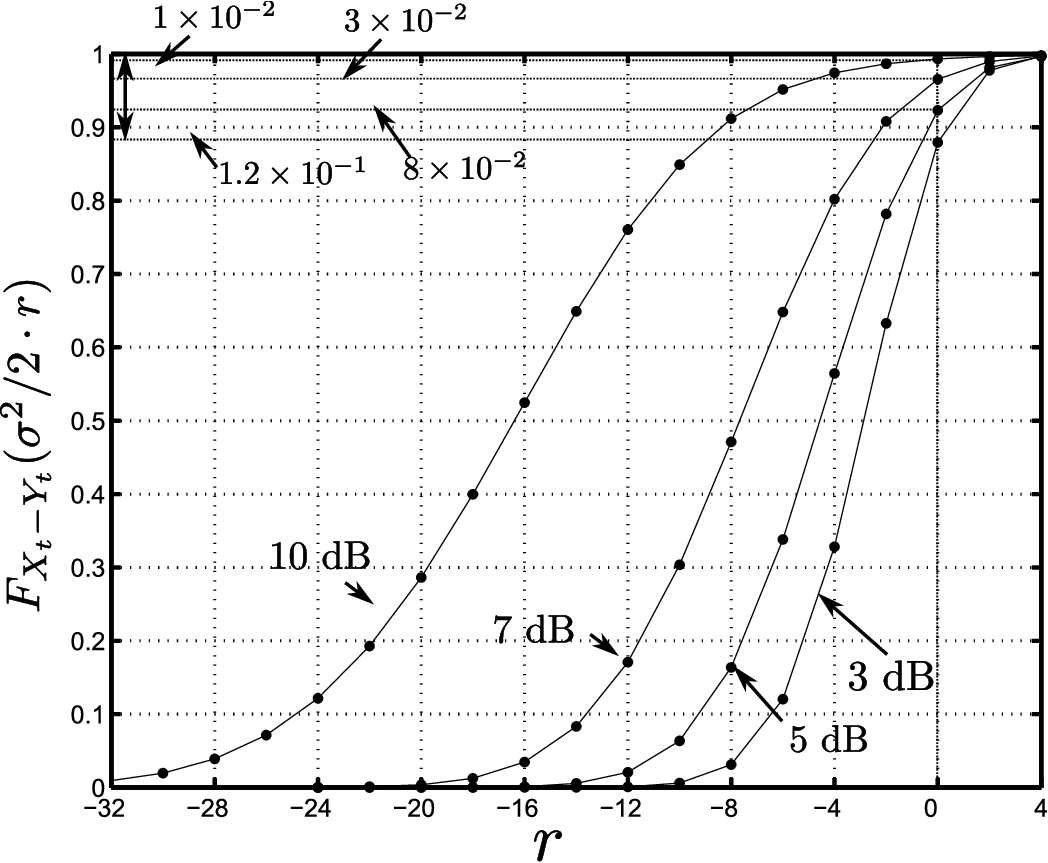}
	\caption{Comparing the distributions $\marg$ across different SNRs, for a fixed truncation length $\L=5$. The channel is the PR1 channel, see Table \ref{tab:1}.}
	\label{fig:marg_comp}
\end{figure}

\newcommand{\joint}{F_{\randb{X}_{\mathbf{t}_1^2}-\randb{Y}_{\mathbf{t}_1^2}}(\sig^2/2 \cdot \mathbf{r})}
\newcommand{\jointT}{F_{\randb{X}_{\mathbf{t}_1^2}-\randb{Y}_{\mathbf{t}_1^2}}(\sig^2/2 \cdot [r_1,r_2]^T)}

\begin{figure*}[t]
	\centering
		\includegraphics[width=.7\linewidth]{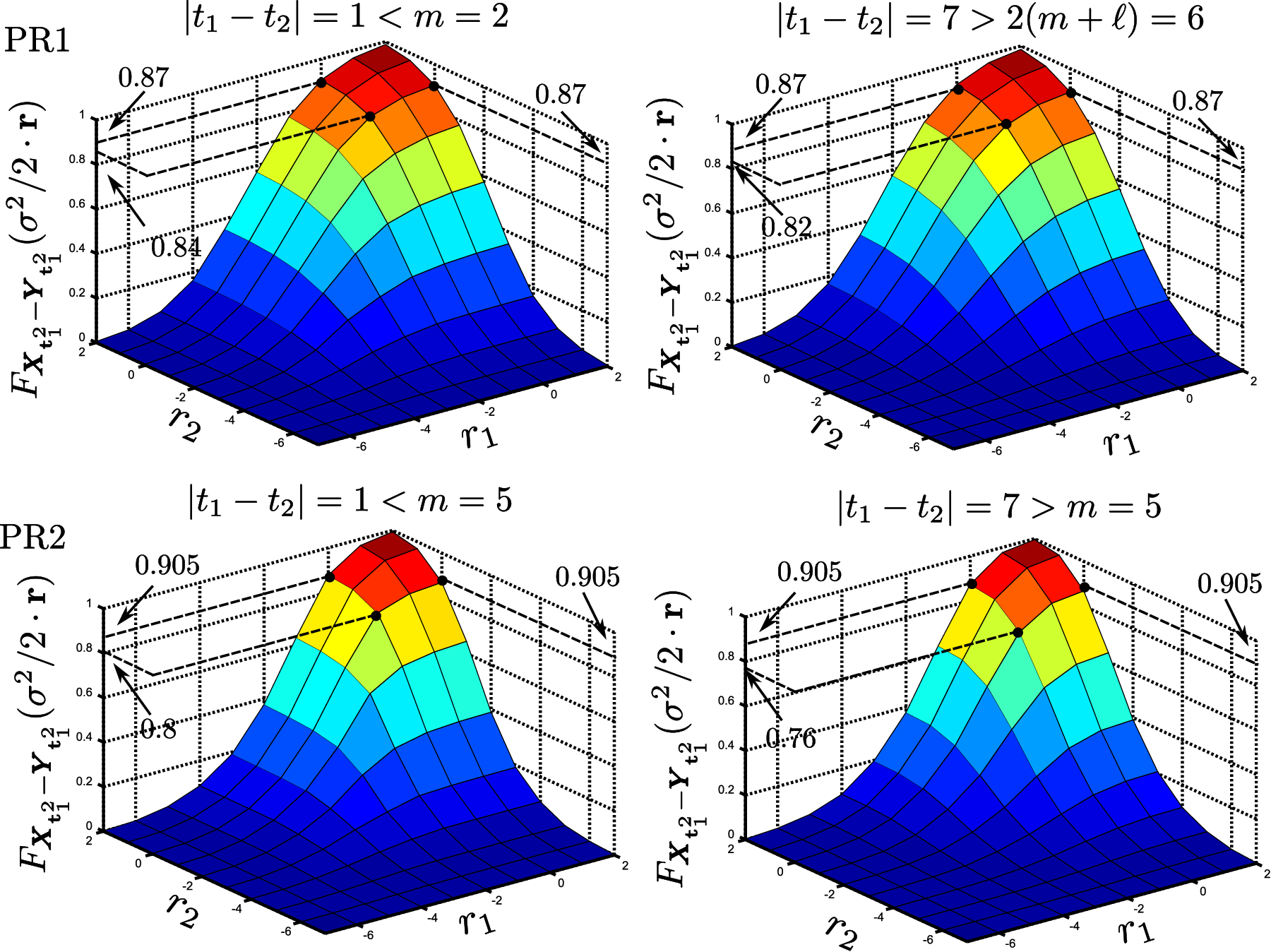}
	\caption{Joint reliability distribution $\joint$ computed for both the PR1 and PR2 channels, with chosen truncation lengths $\L=2$ and $5$. }
	\label{fig:joint_pr2}
\end{figure*}

\subsection{Joint distribution for $n=2$ case, when the noise is i.i.d.}

We consider again i.i.d noise $\rand{W}_t$, and the PR1 and PR2 channels (see Table \ref{tab:1}).
Here, we choose the SNR to be moderate at 5 dB. For the PR1 channel with memory length $\ISI=1$, the truncation length is fixed to be $\L=2$. 
For the PR2 channel with $\ISI=2$, we fix $\L=5$. 
Figure \ref{fig:joint_pr2} compares the joint distributions $\joint$, computed for both PR1 and PR2 channels and for both time lags $|t_1 - t_2 |=1$ (i.e., neighboring symbols) and $|t_1 - t_2 |=7$. 
The difference between the two cases $|t_1 - t_2 |=1$ and $7$ is subtle (but nevertheless inherent) as observed from the differently labeled points in the figure. For the PR1 channel, the joint symbol error probability $\Pr{\rand{B}_{t_1}\neq \rand{A}_{t_1},\rand{B}_{t_2}\neq \rand{A}_{t_2}}= \Pr{\randb{X}_{\mat{t}_1^2} \geq \randb{Y}_{\mat{t}_1^2}} $ is approximately $ 6 \times 10^{-2}$ and $2 \times 10^{-2} $ for both cases $|t_1 - t_2 |=1$ and $7$, respectively. Similarly for the PR2, the (error) probability is approximately $ 3 \times 10^{-2}$ and $1 \times 10^{-2}$ for both respective cases $|t_1 - t_2 |=1$ and $7$. Finally note that for the PR1 channel when $|t_1 - t_2 | = 7 $, both MLM reliability values $R_{t_1} = 2/\sig^2 \cdot |X_{t_1} - Y_{t_1}|$ and $R_{t_2} = 2/\sig^2 \cdot |X_{t_2} - Y_{t_2}|$ are \emph{independent}; this is because then $|t_1 - t_2 | = 7 > 2(\L+\ISI) = 6$, refer to Figure \ref{fig:Trellis}.

\begin{figure*}[t]
	\centering
		\includegraphics[width=.75\linewidth]{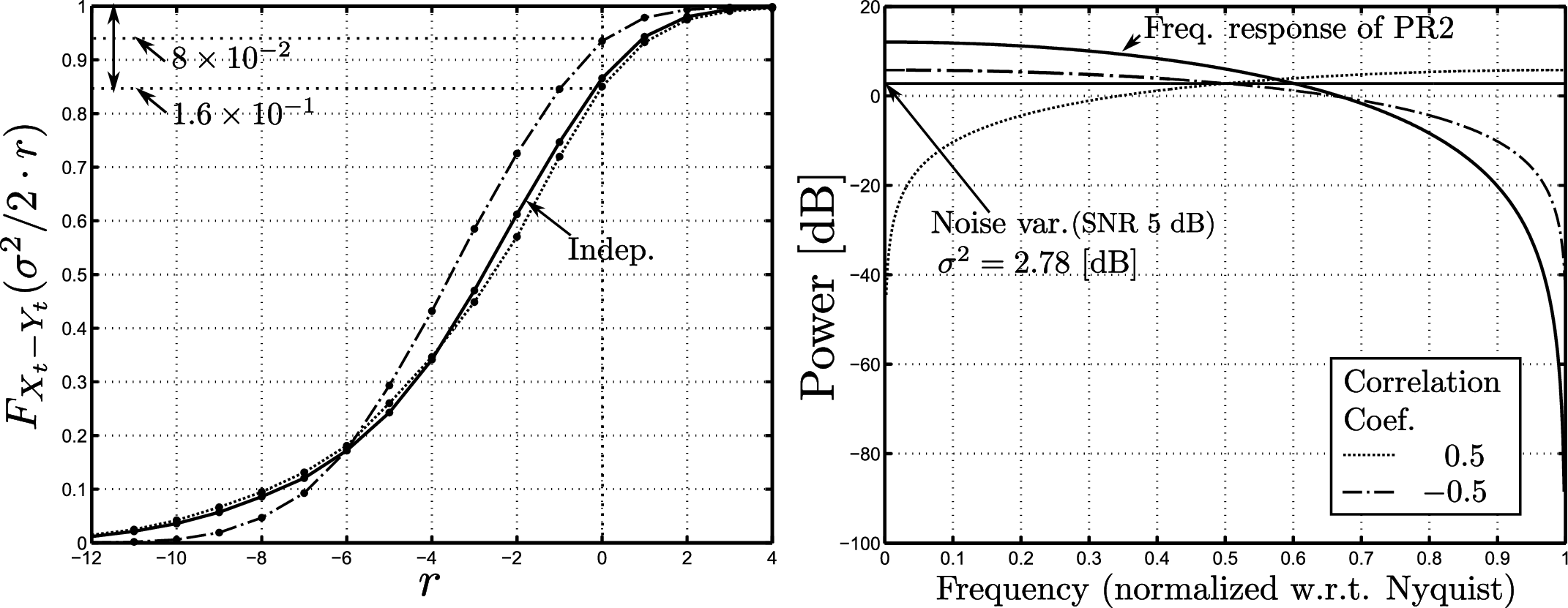}
	\caption{Marginal distribution $\marg$ for correlated noises, for the PR2 channel, at SNR $5$ dB. Truncation length $\L = 5$.}
	\label{fig:marg_corr} 
\end{figure*}

\newcommand{\tp}{\bar{t}}
\newcommand{\corr}{\E \Ev{W_t \cdot W_{t+1}} / \sig^2}

\subsection{Marginal distribution when the noise is correlated.} \label{ssect:corr}

Consider the PR2 channel, and now consider the case where the noise samples $\rand{W}_t$ are \emph{correlated}. For simplicity of argument we consider single lag correlation, i.e. $\E \Ev{W_t \cdot W_{\tp}} = 0$ for all $|t-\tp| > 1$, and consider the following two cases :
\bitm
\item the \emph{correlation coefficient} $\corr = 0.5$, and 
\item the \emph{correlation coefficient} $\corr = -0.5$. 
\eitm
We consider a moderate SNR of 5 dB. Figure \ref{fig:marg_corr} shows the distributions $\marg$ computed for both cases. Also in Figure \ref{fig:marg_corr}, the \emph{power spectral densities} of the correlated noise samples $\rand{W}_t$ (see~\cite{Papoulis:Text}, p.~408) are shown for both cases. It is apparent that the truncated MLM detector performs better (i.e., smaller symbol error probability) when the correlation coefficient $\corr = -0.5$. This is explained intuitively as follows. The detector should be able to tolerate more noise in the signaling frequency region. Observe the PR2 \emph{frequency response}~\cite{PR,Immink} displayed in Figure \ref{fig:marg_corr}. When the correlation coefficient equals $\corr = -0.5$, the noise power is strongest amongst signaling frequencies, and the symbol error probability $\Pr{\rand{B}_t \neq \rand{A}_t}= \Pr{\rand{X}_t \geq \rand{Y}_t}$ is observed to be the lowest (approximately $ 8 \times 10^{-2}$). On the other hand when the correlation coefficient is $\corr = 0.5$, the noise is strongest at frequencies near the \emph{spectral null} of the PR2 channel, and the (error) probability $\Pr{\rand{X}_t \geq \rand{Y}_t} $ is the highest (approximately $ 1.6 \times 10^{-1}$). Note that in the latter case $\corr = -0.5$, the MLM performs even better than the i.i.d case, see Figure \ref{fig:marg_corr}. In the i.i.d case, the error probability $\Pr{\rand{X}_t \geq \rand{Y}_t} \approx 1.3 \times 10^{-1}$.

\begin{rem}
One intuitively expects that similar observations will be made even for other (more complicated) choices for the noise covariance matrix $\Kw$, recall (\ref{eqn:mx_cov}). We stress that our results are general in the sense that we may arbitrarily specify $\Kw$; even if the noise samples $W_t$ are \textbf{non-stationary} our methods still apply.
\end{rem}

\begin{figure*}[!t]
	\centering
		\includegraphics[width=.8\linewidth]{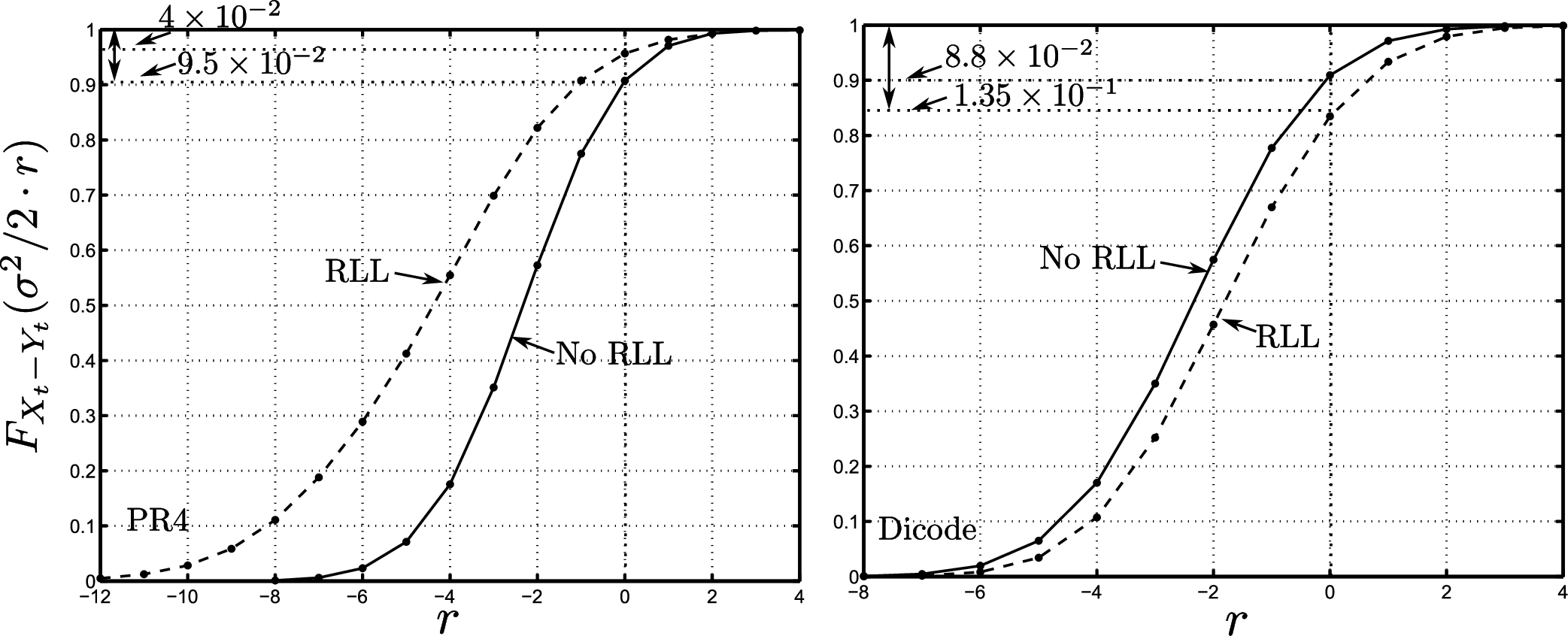}
		\caption{Marginal distributions $\marg$ computed for cases when a run-length limited (RLL) code is present. Here, we compare both the PR4 and dicode (see Table \ref{tab:1}) channels at SNR 5 dB.}
	\label{fig:marg_rll}
\end{figure*}

\subsection{Marginal distribution when the noise is i.i.d., and when run-length limited (RLL) codes are used.} \label{ssect:RLL}

We demonstrate Procedure \ref{proce:FT} in Subsection \ref{ssect:33}, used to compute the distribution $\marg$ when a modulation code is present in the system. In particular, consider a \emph{run-length limited (RLL)} code; we test the simple RLL code that prevents neighboring symbol transitions~\cite{RLL,Immink}. This code improves transmission over ISI channels, that have spectral nulls near the Nyquist frequency~\cite{Immink}; one such channel is the PR4, see Table \ref{tab:1}. Figure \ref{fig:marg_rll} shows $\marg$ computed for both the PR4, as well as the dicode channel, see Table \ref{tab:1}. 
The PR4 channel has a spectral null at Nyquist frequency (recall Subsection \ref{ssect:corr}), but the dicode channel does not.

It is clearly seen from Figure \ref{fig:marg_rll} that the RLL code improves the performance when used for the PR4 channel. For the PR4 channel, the distribution $\marg$ appears to concentrate more probability mass over negative values of $X_t - Y_t$ similar to the observations made in Figure \ref{fig:marg_comp} when there is an increase in SNR. The error probability $\Pr{\rand{B}_t\neq\rand{A}_t}=\Pr{\rand{X}_t \geq \rand{Y}_t} = 1 - F_{X_t -Y_t}(0) $ decreases by a factor of 2, dropping from approximately $9.5 \times 10^{-2}$ to $4 \times 10^{-2}$. On the other hand, the RLL code has a negative impact on the performance when applied to the dicode channel. For the dicode channel, $F_{\rand{{X}}_t-{\rand{Y}}_t}(r)$ concentrates more probability mass over positive values of $X_t - Y_t$ (similar to the observations made in Figure \ref{fig:marg_comp} when there is an SNR decrease), and the (error) probability $\Pr{\rand{X}_t \geq \rand{Y}_t} $ increases from approximately $8.8 \times 10^{-2}$ to $1.35\times 10^{-1}$.

\subsection{Marginal distribution when conditioning on neighboring error events}

\newcommand{\Xtt}[1]{\randb{X}_{#1}}
\newcommand{\Ytt}[1]{\randb{Y}_{#1}}
\newcommand{\rr}{r}

Here we consider three \emph{neighboring} symbol reliabilities, i.e., we consider $\randb{R}_{\mat{t}_1^3} = [\rand{R}_{t-1},\rand{R}_t, \rand{R}_{t+1}]^T$. We consider the following two conditional distributions :

\begin{align}
\mbox{(a)}&~~~~\Pr{\rand{X}_t - \rand{Y}_t \leq \rr | \rand{X}_{t-1} < \rand{Y}_{t-1}, \rand{X}_{t+1} < \rand{Y}_{t+1}} \nn &= \frac{1}{\C_1} \cdot  F_{\Xtt{\mat{t}_1^3}-\Ytt{\mat{t}_1^3}}(0,\rr,0),~\mbox{and }\nn
\mbox{(b)}&~~~~\Pr{\rand{X}_t - \rand{Y}_t \leq  \rr | \rand{X}_{t-1} \geq \rand{Y}_{t-1}, \rand{X}_{t+1} \geq \rand{Y}_{t+1}}\nn &= \frac{1}{\C_2} \left(F_{\rand{X}_t-\rand{Y}_t}(\rr) - F_{\Xtt{\mat{t}_2^3}-\Ytt{\mat{t}_2^3}}(\rr,0)\right.  \nn
&~~~~~~~~~~ \left. - F_{\Xtt{\mat{t}_1^2}-\Ytt{\mat{t}_1^2}}(0,r) +  F_{{\Xtt{\mat{t}_1^3}-\Ytt{\mat{t}_1^3}}}(0,r,0) \right), \nonumber
\end{align}
where the normalization constants $\C_1$ and $\C_2$ equal the probabilities of the (respective) events that were conditioned on. 
Distribution (a) is conditioned on the event that both neighboring symbols are \emph{correct}. i.e., $\Ev{\rand{B}_{t-1} = \rand{A}_{t-1}, \rand{B}_{t+1} = \rand{A}_{t+1}}$. Distribution (b) is conditioned on the event that both neighboring symbols are \emph{wrong}, i.e., $\Ev{\rand{B}_{t-1} \neq \rand{A}_{t-1}, \rand{B}_{t+1} \neq  \rand{A}_{t+1}}$. For the PR1, PR2 and PR4 channels, both conditional distributions (a) and (b) are shown in Figures \ref{fig:triple_comp2} and \ref{fig:triple_comp}. We compare two different SNRs 3 and 10 dB. For comparison purposes, we also show the \emph{unconditioned} distribution $\dist$ in both Figures \ref{fig:triple_comp2} and \ref{fig:triple_comp}. We make the following observations.

In all considered cases, distribution (a) is seen to be similar to the {unconditioned} distribution. However, distribution (b) is observed to vary for all the considered cases. Take for example the PR2 channel, we see from Figure \ref{fig:triple_comp} that distribution (b) has probability mass concentrated to the right of the unconditioned $\marg$. 
This is true for both SNRs 3 and 10 dB.
In contrast for the PR1, the MLM detector behaves differently at the two SNRs. We see from Figure \ref{fig:triple_comp2} that at SNR 10 dB, the distribution (b) has a lower symbol error probability than that of the unconditioned $\marg$. At SNR 3 dB however, the opposite is observed, i.e., the symbol error probability is higher than that of the distribution $\marg$. This is because at SNR 10 dB, errors occur \emph{sparsely}, interspaced by correct symbols; it is uncommon to encounter \emph{consecutive} symbols in error. Hence conditioned on adjacent symbols $B_{t-1}$ and $B_{t+1}$ being wrong, it is uncommon for $B_t$ to be also wrong, as this is the event where we have three consecutive erroneous symbols. 
Finally, the observations made for the PR4 channel are again different. 
We notice that both distributions (a) and (b) practically equal the unconditioned distribution $\marg$. This is because the even/odd output subsequences of the PR4 channel are independent of each other.

\begin{figure}[!t]
	\centering
		\includegraphics[width=.9\linewidth]{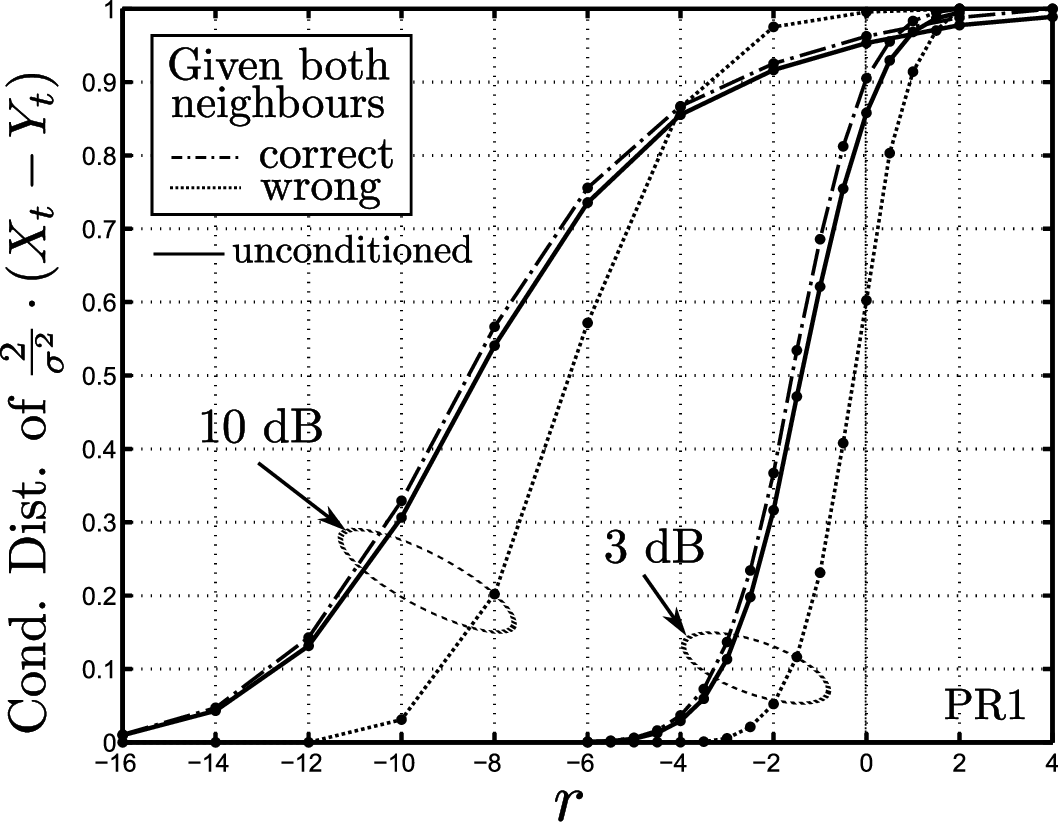}
	\caption{Marginal distributions of ${\rand{X}_t-\rand{Y}_t}$ computed for the PR1 channel, obtained when conditioning on either events $\Ev{\rand{B}_{t-1} \neq \rand{A}_{t-1}, \rand{B}_{t+1} \neq \rand{A}_{t+1}}$ and $\Ev{\rand{B}_{t-1} = \rand{A}_{t-1}, \rand{B}_{t+1} = \rand{A}_{t+1}}$. These two events correspond to error (or non-error) events at neighboring time instants $t-1$ and $t+1$. The solid black line represents the unconditioned marginal distribution of ${\rand{X}_t-\rand{Y}_t}$. }
	\label{fig:triple_comp2}
\end{figure}

\begin{figure*}[!t]
	\centering
		\includegraphics[width=.8\linewidth]{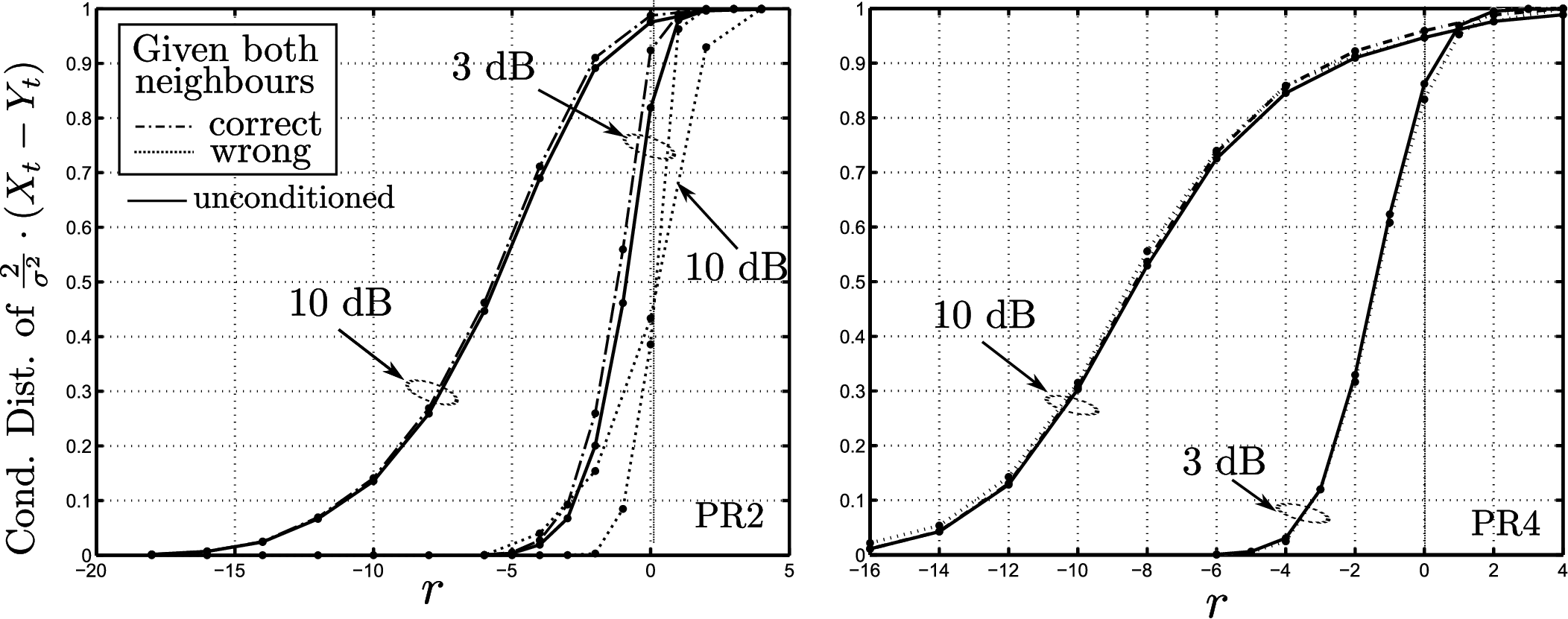}
		\caption{Marginal distributions of ${\rand{X}_t-\rand{Y}_t}$ computed for both the PR2 and PR4 channels, obtained when conditioning on either events $\Ev{\rand{B}_{t-1} \neq \rand{A}_{t-1}, \rand{B}_{t+1} \neq \rand{A}_{t+1}}$ and $\Ev{\rand{B}_{t-1} = \rand{A}_{t-1}, \rand{B}_{t+1} = \rand{A}_{t+1}}$. These two events correspond to error (or non-error) events at neighboring time instants $t-1$ and $t+1$. The solid black line represents the unconditioned marginal distribution of $2/\sig^2 \cdot (\rand{X}_t-\rand{Y}_t)$. }
	\label{fig:triple_comp}
\end{figure*}

\section{Conclusion} \label{sect:con}

In this paper for the $\L$-truncated MLM detector, we derived closed-form expressions for both i) the reliability distributions $\dist$, and ii) the symbol error probabilities $\PSE$. Our results hold jointly for any number $n$ of arbitrarily chosen time instants $t_1,t_2,\cdots, t_n$. The general applicability of our result has been demonstrated for a variety of scenarios. Efficient Monte-Carlo procedures that utilize dynamic programming simplifications have been given, that can be used to numerically evaluate the closed-form expressions. 

It would be interesting to further generalize the exposition to consider \emph{infinite impulse response} (IIR) filters, such as in \emph{convolutional codes}.

\section*{Acknowledgment}
The authors would like to thank the Associate Editor, and the anonymous reviewers, for their help in greatly improving the presentation of our initial draft.

\appendix

\subsection{Computing the matrix $\mathbf{Q}= \mathbf{Q}(\textbf{A}_{\mathbf{t}_1^n})$ in Definition \ref{defn:QD}} \label{app:spec}

\renewcommand{\m}{n}
\renewcommand{\F}{\pmb{\mathcal{C}}}

In this appendix, we show that the size $2mn$ square matrix $\Q$ with both properties i) and ii) as stated in Definition \ref{defn:QD}, can be easily found. We begin by noting from (\ref{eqn:SS}) that $\rank(\S\S^T)=2m$, therefore the matrix $\I_n \otimes \S\S^T$ has rank $2mn$ and is \emph{positive definite}. Recall $\bdiag\left( \GMats \right)$ is {block diagonal } with entries~(\ref{eqn:G}).

\begin{lem} \label{lem:SD}
Let $\S$ be given as in (\ref{eqn:S}). Let the size $2mn$ by $2mn$ square matrix $\Alp$ diagonalize
\bea
	\Alp^T (\I_n \otimes \S\S^T) \Alp = \I. \label{eqn:SD}
\eea
Let $\Beta$ be a $2mn \times 2mn$ eigenvector matrix $\Beta$ in the following decomposition
\begin{align}
\renewcommand{\arraystretch}{.7}
&\Alp^T(\I_n \otimes \S\S^T) \bdiag\left(\GMats\right)^T
\Kw \nn &~\cdot~
\renewcommand{\arraystretch}{.7}
\bdiag\left(\GMats\right) (\I_n \otimes \S\S^T)\Alp \nn
&=  \Beta\Del^2 \Beta^T, \label{eqn:Fmat} 
\end{align}
where $\Del^2$ is the eigenvalue matrix of (\ref{eqn:Fmat}), therefore $\Del^2$ in (\ref{eqn:Fmat}) is diagonal of size $2mn$. Then
\bea
	\Q = \Alp \Beta
\eea
satisfies both properties i) and ii) stated in Definition \ref{defn:QD}.
\end{lem}
\begin{proof}
Because $\Alp$ diagonalizes $\I_n \otimes \S\S^T$ to an identity matrix $\I$, it follows that $\Alp$ must have full rank, and thus have an inverse $\Alp^{-1}$. It follows from (\ref{eqn:SD}) that $\Alp^{-1} = \Alp^T(\I_n\otimes \S\S^T)$. Replacing $\Alp^T(\I_n\otimes \S\S^T) = \Alp^{-1}$ in (\ref{eqn:Fmat}), we see that $\Beta$ satisfies
\begin{align}
&\Alp^{-1} \bdiag\left(\GMats\right)^T
\Kw \nn ~&~\cdot
\bdiag\left(\GMats\right) \Alp^{-T} =  \Beta\Del^2 \Beta^T. \label{eqn:Fmat2} 
\end{align}
Consider the matrix $\Q=\Alp\Beta$. It follows from (\ref{eqn:Fmat2}) that $\Q=\Alp\Beta$ satisfies property i) in Definition \ref{defn:QD}, as seen after multiplying (the matrices satisfying) (\ref{eqn:Fmat2}) on the left and right by $\Alp$ and $\Alp^T$, respectively. It also follows that $\Q=\Alp\Beta$ satisfies property ii) in Definition \ref{defn:QD}, this is because
\[
\Q^T (\I_n\otimes \S\S^T) \Q = \Beta^T \Alp^T(\I_n\otimes \S\S^T) \Alp \Beta  = \Beta^T \Beta = \I, 
\] 
where the last equality follows because $\Beta$ is \emph{unitary} (i.e., $\Beta^{-1} = \Beta^T$) by virtue of the fact that it is an eigenvector matrix~\cite{Golub}, p.~311. 
\end{proof}

To summarize Lemma \ref{lem:SD}, the matrix $\Q= \Q(\Atk)$ in Definition \ref{defn:QD}, is obtained by first computing two size $2mn$ matrices $\Alp$ and $\Beta$ respectively satisfying (\ref{eqn:SD}) and (\ref{eqn:Fmat}), and then setting $\Q= \Alp\Beta$. The matrix $\Beta$ is obtained from an eigenvalue decomposition of the $2mn$ matrix (\ref{eqn:Fmat}), and clearly $\Beta$ depends on the symbols $\Atk$. The matrix $\Alp$ however, is simpler to obtain. This is due to the simple form of $\S\S^T$ in (\ref{eqn:SS}), and we may even obtain closed form expressions for $\Alp$, see the next remark.

\begin{rem}
It can be verified that the following are eigenvectors of the matrix $\S\S^T$ in (\ref{eqn:SS}). The first $2m-1$ eigenvectors are
\bea
(i + i^2)^{-\frac{1}{2}} \cdot [\overbrace{1, 1, \cdots, 1}^{i}, -i, \overbrace{0, 0, \cdots, 0}^{2m-(i+1)}]^T \nonumber 
\eea
where $i$ can take values $1 \leq i < 2m$, and the last eigenvector is simply $\1/|\1|= \1/(2m)$.  
\end{rem}

\bibliographystyle{IEEEtran}

}

\begin{IEEEbiographynophoto}{Fabian Lim}
received the B.Eng and M.Eng degrees from the National University of Singapore in 2003 and 2006, respectively, and the Ph.D. degree from the University of Hawaii, Manoa in 2010, all in  electrical engineering. Currently, he is a postdoctoral associate at the Massachusetts Institute of Technology. 

Dr. Lim has held short-term visiting research positions at Harvard University in 2004 and 2005. From Oct 2005 to May 2006, he was a staff member in the Data Storage Institute, Singapore. From May 2008 to July 2008, he as an intern at Hitachi Global Storage Technologies, San Jose. In March 2009, he was a visitor at the Research Center for Information Security, Japan. 

His research interests include error-control coding and signal processing.
\end{IEEEbiographynophoto}

\begin{IEEEbiographynophoto}{Aleksandar Kav\v{c}i\'c} received the Dipl. Ing. degree in Electrical Engineering from Ruhr-University, Bochum, Germany in 1993, and the Ph.D. degree in Electrical and Computer Engineering from Carnegie Mellon University, Pittsburgh, Pennsylvania in 1998.

Since 2007 he has been with the University of Hawaii, Honolulu where he is presently Professor of Electrical Engineering. Prior to 2007, he was in the Division of Engineering and Applied Sciences at Harvard University, as Assistant Professor and Associate Professor of Electrical Engineering. He also served as Visiting Associate Professor at the City University of Hong Kong in the Fall of 2005 and as Visiting Scholar at the Chinese University of Hong Kong in the Spring of 2006.

Prof. Kav\v{c}i\'c received the IBM Partnership Award in 1999 and the NSF CAREER Award in 2000. He is a co-recipient, with X. Ma  and N. Varnica, of the 2005 IEEE Best Paper Award in Signal Processing and Coding for Data Storage. He served on the Editorial Board of the IEEE Transactions on Information Theory as Associate Editor for Detection and Estimation from 2001 to 2004, as Guest Editor of the IEEE Signal Processing Magazine in 2003-2004, and as Guest Editor of the IEEE Journal on Selected Areas in Communications in 2008-2009. From 2005 until 2007, he was the Chair of the Data Storage Technical Committee of the IEEE Communications Society.
\end{IEEEbiographynophoto}

\end{document}